\documentclass[leqno,12pt]{article}

\usepackage{amsmath}
\usepackage{amssymb}
\usepackage[authoryear]{natbib}
\usepackage{graphicx}
\usepackage{epsfig}
\usepackage{color}
\usepackage{setspace}

\makeatletter\@addtoreset{equation}{section}\makeatother

\newcommand{\R}{\mathbb{R}}
\newcommand{\N}{\mathbb{N}}

\newcommand{\bea}{\begin{eqnarray*}}
\newcommand{\eea}{\end{eqnarray*}}
\newcommand{\be}{\begin{eqnarray}}
\newcommand{\ee}{\end{eqnarray}}
\newcommand{\beq}{\begin{equation}}
\newcommand{\eeq}{\end{equation}}
\newtheorem{assumption}{Assumption}[section]

\newtheorem{thm}{Theorem}[section]
\newtheorem{cor}[thm]{Corollary}
\newtheorem{lemma}[thm]{Lemma}

\newtheorem{rem}[thm]{Remark}

\newtheorem{algorithm}[thm]{Algorithm}

\def\3{\ss}

\DeclareMathOperator*{\arginf}{arg\,inf}
\DeclareMathOperator*{\argmax}{arg\,max}

\begin{document}

\title{Bayesian $T$-optimal discriminating designs }

\author{
{\small Holger Dette} \\
{\small Ruhr-Universit\"at Bochum} \\[-6pt]
{\small Fakult\"at f\"ur Mathematik} \\[-6pt]
{\small 44780 Bochum, Germany} \\[-6pt]
{\small e-mail: holger.dette@rub.de}\\
{\small Viatcheslav B. Melas, Roman Guchenko } \\[-6pt]
\small St. Petersburg State University \\[-6pt]
\small Department of Mathematics \\[-6pt]
\small St. Petersburg,  Russia \\[-6pt]
{\small e-mail: vbmelas@post.ru,romanguchenko@ya.ru}}
\maketitle

\begin{abstract}
The problem of constructing Bayesian optimal discriminating designs for a class of regression models
with respect to  the
$T$-optimality criterion  introduced by \cite{atkfed1975a} is considered. It is demonstrated that the discretization
 of the integral with respect to the prior distribution leads to locally $T$-optimal discriminating design  problems with
 a large number of model comparisons. Current methodology for the numerical construction of discrimination designs
  can only deal with a few comparisons, but the discretization of the Bayesian prior easily yields to discrimination design problems for
 more than $100$  competing models.  A new  efficient  method is developed to deal with problems of this type.
It combines some features of the classical exchange type algorithm  with the gradient methods.
 Convergence is proved and it is demonstrated that
 the new method can find Bayesian optimal discriminating designs in situations where all currently available procedures fail.
\end{abstract}

\medskip



\section{Introduction} \label{sec1}

Optimal design theory provides useful tools to improve the accuracy of statistical inference without any additional costs by carefully
planning experiments before they are conducted. Numerous authors have worked on the construction of optimal
designs in various situations. For many models optimal designs have been developed explicitly
[see the monographs of  \cite{pukelsheim2006,atkinson2007}]
 and several algorithms
have been developed for their numerical construction if the optimal designs are not available in explicit form [see
\cite{yu2010,yanbie2013} among others]. On the other hand the construction of such
designs depends sensitively on the model assumptions and an optimal  design for a  particular model might be inefficient
if it is used in a different model. Moreover, in many experiments it is often not obvious which model should be finally  fitted to the data
and model building is an important part of data analysis. A typical  and very important
example are Phase II dose-finding studies, where various nonlinear regression  models
of the form
\be
  Y=\eta(x,\theta)+\varepsilon. \label{1.1}
  \ee
have been developed for describing the dose-response
relation  [see \cite{pinbrebra2006}], but the problem of model uncertainty arises in nearly any other statistical application.
As a consequence, the construction
 of efficient designs for model identification has become an important field in optimal design theory.
Early work can be found in \cite{stigler1971}, who determined 
designs for discriminating between
two nested univariate polynomials by  minimizing  the volume of the confidence
ellipsoid for the parameters corresponding to the extension of the smaller model.
Several authors have   worked on this approach in various other classes of nested models
[see for example \cite{dethal1998} or  \cite{songwong1999} among others]. \\
A different approach to the problem of constructing optimal designs for model discrimination
is given in a pioneering paper  by \cite{atkfed1975a}, who proposed the $T$-optimality criterion
to construct designs for discriminating between two competing regression models.
Roughly speaking their approach provides a design such that the sum of squares
for a lack of fit test is large. \cite{atkfed1975b} extended this method
for discriminating a selected model $\eta_1$ from
a class of other regression models,
say $ \{\eta_2, \ldots , \eta_k \}$, $k \ge 2$. In contrast to  the work \cite{stigler1971} and followers
the $T$-optimality criterion  does not require competing nested models and has found considerable attention
in the statistical literature with   numerous applications including such important fields as chemistry
or pharmacokinetics
 [see e.g.\
  \cite{atkbogbog1998}, \cite{ucibog2005},
 \cite{loptomtra2007}, \citet{atkinson2008}, \cite{Tommasi09} or \cite{fooduf2011}
 for some more recent references]. A drawback of the $T$-optimality criterion consists of the fact that -- even in the case of linear models -- the criterion depends on the parameters of the model $\eta_1$.
 This means that $T$-optimality is a local optimality criterion in the sense of \cite{chernoff1953}, and that it requires some preliminary knowledge regarding the parameters. Consequently,
 most of the cited papers refer to locally $T$-optimal
 designs. Although there exist applications where such information is
 available [for example in the analysis of dose response studies as considered in  \cite{pinbrebra2006}], in most
 situations such knowledge can be rarely provided.  Several authors have introduced robust
 versions of the classical optimality criteria
  such as Bayesian or minimax $D$-optimality criteria in order to determine efficient designs for model discrimination,
 which are less sensitive with respect to the choice of parameters
 [see \cite{pronwalt1985,chaver1995,dette1997}].  The robustness problem of the $T$-optimality criterion
  has been already mentioned in
 \cite{atkfed1975a}, who proposed a Bayesian approach  to address the problem of parameter uncertainty in the $T$-optimality criterion.
    \cite{wiens2009} imposed (linear)  neighbourhoud
structures on each regression response and determined  least favorable
 points in  these neighbourhouds in order to robustify the locally $T$-optimal design problem.
  \cite{detmelshp2012}  considered polynomial regression models and determined explicitly Bayesian $T$-optimal
 discriminating  designs for the criterion introduced by  \cite{atkfed1975a}. Their results indicate the difficulties arising in Bayesian $T$-optimal design problems.  \\
 The scarcity of
  literature on  Bayesian
 $T$-optimal discriminating designs can be explained by the fact that in nearly all cases
 of practical interest these designs have to be found numerically, and even this  is a very hard problem.
These numerical difficulties become even apparent in the case   of locally $T$-optimal designs.
  \cite{atkfed1975a} proposed an exchange type algorithm, which has a rather slow rate of convergence and
  has been used by several authors. \cite{BraessDette2013} pointed out that, besides its slow convergence,
    this algorithm does not yield the
  solution of the optimal discriminating design problem, if more than $5$ model comparisons are under consideration.
  These authors developed  a more efficient algorithm for the determination of locally $T$-optimal discriminating designs
  for several competing regression models by exploring relations between optimal design problems
and  (nonlinear) vector-valued approximation theory.  Although the resulting algorithm provides a
substantial improvement of the exchange type methods it cannot deal with Bayesian optimality criteria  in general,
and the development of an efficient procedure for this purpose is a very challenging and open  problem. \\
The goal of the present paper is to fill  this gap. We utilize the fact
that in applications the integral  with respect to the prior distribution has to be determined by
a discrete approximation  and  we show that the discrete Bayesian $T$-optimal design problem
 is  a special case of the local $T$-optimality criterion for a very large number of competing
 models considered as in \cite{BraessDette2013}. The competing models arise from the different support points
 used for the approximation of the prior distribution by a discrete measure, and  the number of
model comparisons in the resulting criterion  easily exceeds the $200$. Therefore
the algorithm in  \cite{BraessDette2013} does not provide a 
 solution of the corresponding optimization problem, and we propose a new method for the numerical construction of
  Bayesian $T$-optimal designs
 with   substantial computational advantages.
Roughly speaking, the support points of the design in each iteration
are determined in a similar manner as proposed  in   \cite{atkfed1975a} but for the calculation of the
 corresponding  weights  we use a gradient approach. It turns out that the new procedure is extremely
efficient and is able to find Bayesian $T$-optimal designs with a few number of iterations. \\
 The remaining part of this paper is organized as follows. In Section \ref{sec2} we give an introduction into
 the problem of designing experiments for discriminating between  competing regression models and also derive some basic properties
 of locally $T$-optimal discriminating designs. In particular we show how the Bayesian $T$-optimal design
 problem is related to a local one with a large number of model comparisons [see Section \ref{sec2b}].
 Section \ref{sec3} is devoted to the construction of new numerical procedures (in particular Algorithm  \ref{algorithm:new}),
 for which we prove  convergence  to a $T$-optimal discriminating design. Our approach consists of
 two steps  consecutively optimizing with respect to the support points (Step 1) and weights of the design (Step 2).
For the second step we  also discuss two procedures to speed up the convergence of the algorithm.
The results are illustrated in Section \ref{sec5} calculating several Bayesian
$T$-optimal discriminating designs in examples, where all other available procedure do not provide a numerical solution
of the optimal design problem. For example, the new procedure is able to solve locally $T$-optimal designs with more
than $240$ model comparisons as they are arising frequently in Bayesian $T$-optimal design problems. In particular
we illustrate the methodology calculating Bayesian $T$-optimal discriminating designs for a dose finding clinical trial
which has recently been discussed in  \cite{pinbrebra2006}. The corresponding R-package will be provided in the CRAN library.
 Finally all proof are deferred to an appendix in Section \ref{sec6}.

\section{$T$-optimal discriminating designs} \label{sec2}

Consider the regression model \eqref{1.1}, where $x$ belongs to some compact set $\mathcal{X}$ and observations at different experimental conditions are
independent.  For the sake of transparency and a clear representation we assume that the error $\varepsilon$
is  normally distributed. The methodology developed in the following discussion can be extended to more general error structures following
 the line of research  in   \cite{loptomtra2007}, but details are omitted for the sake of brevity.
  
  Throughout this paper we consider the situation
where    $\nu $ different models, say
\begin{align} \label{2.1}
\eta_i(x,\theta_{i}), \qquad i = 1,\dots,\nu,
\end{align}
are available to describe the dependency of $Y$ on the predictor $x$.
In \eqref{2.1} the quantity  $\theta_{i}$ denotes  a $d_i$-dimensional parameter, which varies in a compact
space, say $\Theta_i$ ($i=1,\ldots , \nu$).
Following \cite{kiefer1974} we consider approximate
designs that are defined as probability measures, say $\xi$,  with finite support.
 The support points $ x_1,\ldots, x_k $ of a design $\xi$ give the locations
 where observations are taken, while the weights $\omega_1,\ldots, \omega_k$ describe the
 relative proportions of observations at these points. If an approximate design is given and  $n$
 observations can be taken, a rounding procedure is applied
to obtain  integers $n_{i} $ ($i=1,\ldots,k)$  from the not necessarily integer valued quantities
$\omega_{i}n$ such that $\sum_{i=1}^kn_i=n$.
  We are interested in designing an
 experiment, such that a most appropriate model  can be  chosen from the given class
 $\{\eta_1, \ldots , \eta_\nu\}$ of competing models.

\subsection{$T$-optimal designs} \label{sec2a}

In the case of $\nu =2$ competing models  \cite{atkfed1975a} proposed to fix one model, say $\eta_1 (\cdot, \theta_1)$,  with corresponding parameter $\overline \theta_1$ and to maximize
 the functional
\begin{align} \label{2.2}
T_{12}(\xi) = \inf_{\theta_{2} \in \Theta_2} \int_{\mathcal{X}} \Big[ \eta_1(x,\overline{\theta}_{1}) - \eta_2(x,\theta_{2}) \Big]^2 \xi(dx),
\end{align}
in the class of all (approximate) designs. Roughly speaking, these designs maximize the power of the test of the
hypothesis ''$\eta_1$ versus $\eta_2$'' . Note that the resulting optimal design depends on the parameter   $\overline{\theta}_1$
for the first model, which has to be fixed by the experimenter. This means that these designs are local in the sense of \cite{chernoff1953}. It was pointed out by \cite{DetteMelasShpilev2013} that locally $T$-optimal designs may be very sensitive with respect to misspecification of $\overline \theta_1$.
In a further paper \cite{atkfed1975b} generalized their approach to construct optimal
 discriminating designs for more than $2$   competing  regression models and suggested the criterion
 \begin{align} \label{2.3}
T (\xi) = \min_{2 \leq j \leq \nu} T_{1j}(\xi)  =
\min_{2 \leq j \leq \nu}\inf_{\theta_{j} \in \Theta_j} \int_{\mathcal{X}} \Big[ \eta_1(x,\overline{\theta}_{1}) - \eta_j(x,\theta_{j}) \Big]^2 \xi(dx).
\end{align}
This criterion determines a ''good'' design for discriminating the model $\eta_1$ against $\eta_2, \ldots, \eta_\nu$, where
the parameter  $\overline{\theta}_1$ has the same meaning as before.  As  pointed out by \cite{tomlop2010} and
 \cite{BraessDette2013} there are many situations, where it is not clear which model should be considered as fixed and
 these authors proposed a symmetrized Bayesian (instead of minimax) version of the $T$-optimality criterion, that is
\begin{align} \label{2.4}
T_{\mathrm{P}}(\xi)  = \sum_{i,j=1}^{\nu} p_{i,j} T_{i,j}(\xi )  
  =  \sum_{i,j=1}^{\nu} p_{i,j} \inf_{\theta_{i,j} \in \Theta_j} \int_{\mathcal{X}} \Big[ \eta_i(x,\overline{\theta}_{i}) - \eta_j(x,\theta_{i,j}) \Big]^2 \xi(dx),
\end{align}
where the quantities $p_{i,j}$ denote nonnegative weights reflecting the importance of the comparison between the
the model $\eta_i$ and $\eta_j$.  We note again that this criterion requires the specification of the parameter $\overline{\theta}_{i}$,
whenever the corresponding weight  $p_{i,j}$ is positive.   Throughout this paper we will call a design maximizing one of the
criteria \eqref{2.2} - \eqref{2.4} locally $T$-optimal  discriminating design, where the specific criterion  under consideration is always clear
from the context. For some recent references discussing locally $T$-optimal discriminating designs we refer to
\cite{ucibog2005},
 \cite{loptomtra2007}, \citet{atkinson2008},  \cite{Tommasi09} or  \cite{BraessDette2013}  among many others. For the formulation of the first results we require the following assumptions.

 \begin{assumption}
\label{assum1}
For each   $i=1,\dots,\nu $ the functions  $\eta_i(\cdot ,\theta_{i})$ is  continuously differentiable with respect to the  parameter $\theta_{i} \in \Theta_i, $ .
\end{assumption}
\begin{assumption}\label{assum2}
For any design $\xi$   such that
 $T_\mathrm{P}(\xi)>0$  and weight $p_{i,j} \neq 0$ the infima in \eqref{2.4} are attained at a unique points  $ \widehat{\theta}_{i,j}  = \widehat{\theta}_{i,j}(\xi)$
in the interior of the set $\Theta_j$.
\end{assumption}

For a design $\xi$ we also introduce the notation
\begin{align}
{\Theta}_{i,j}^*(\xi)  =   \arginf_{\theta_{i,j} \in \Theta_j} \int_{\mathcal{X}}
 \big[ \eta_i(x,\overline{\theta}_i) - \eta_j(x, \theta_{i,j}) \big]^2 \xi(dx) ,
\label{thetamin}
\end{align}
which is used in the formulation of
 the following result.

\begin{thm} \label{thm1}
If Assumption \ref{assum1} is satisfied, then the design  $\xi^*$ is a locally  $T_\mathrm{P}$-optimal discriminating design, if and only if there
exist    distributions $\mu_{ij}^*$  on the  sets $ {\Theta}_{i,j}^*(\xi^*)$ defined in  \eqref{thetamin} such that the
inequality
\begin{equation} \label{equiv}
 \sum_{i,j = 1}^{\nu} p_{i,j}  \int_{{\Theta}_{i,j}^*(\xi^*)}  \big[ \eta_i(x,\overline{\theta}_{i}) - \eta_j(x,{\theta}_{i,j}) \big]^2 \mu_{ij}^* (d {\theta}_{i,j})
 ~\leq  T_{\mathrm{P}}(\xi^*)
 \end{equation}
 is satisfied for all
$ x \in \mathcal{X}$. Moreover, there is equality in \eqref{equiv}
 for all support points  of the  the  locally  $T_\mathrm{P}$-optimal discriminating design  $\xi^*$.
\end{thm}

Theorem \ref{thm1} provides an extension of the corresponding theorem in ~\cite{BraessDette2013}, and the proof is similar and therefore omitted.
For  designs $\xi, \zeta$  on $\mathcal {X}$  we introduce the function
\begin{align} \label{qfct}
Q(\zeta,\xi)  = \int_\mathcal{X} \sum_{i,j = 1}^{\nu} p_{i,j}\inf_{\theta_{i,j} \in  \Theta_{ij}^*(\xi)}  \Big[ \eta_i(x,\overline{\theta}_{i}) - \eta_j(x,{\theta}_{i,j}) \Big]^2 \zeta (dx) ,
\end{align}
where $\zeta$ is an experimental design and the set $\Theta_{ij}^*(\xi)$ is defined in \eqref{thetamin}.  Using Lemma \ref{lemma1}  from the appendix it is easy to check that
\begin{align*}
\frac{\partial T_{\mathrm{P}}(\xi(\alpha))}{\partial\alpha}\Big|_{\alpha=0}=
Q(\zeta, \xi)-T_{\mathrm{P}}(\xi) 
\end{align*}
where $
\xi(\alpha)=(1-\alpha) \xi + \alpha \zeta $ denotes the  convex combination of the designs $\xi$ and  $\zeta$. If Assumption \ref{assum2} is satisfied, the function $Q$ simplifies to
\begin{align} \label{psi}
Q(\zeta,\xi)  &=  \int_{\mathcal{X}} \sum_{i,j = 1}^{\nu} p_{i,j}   \big[ \eta_i(x,\overline{\theta}_{i}) - \eta_j(x,  \widehat{\theta}_{i,j}  ) \big]^2 \zeta  (dx) , \nonumber
\end{align}
which  plays an important role in the subsequent discussion. In particular we need also the following extension of Theorem \ref{thm1}.

\begin{thm}\label{thm2.2}
If Assumption \ref{assum1} is satisfied and the design  $\xi$ is not $T_\mathrm{P}$-optimal, then there exists a design $\zeta^*$, such that the inequality
$Q(\zeta^*,\xi)> T_{\mathrm{P}}(\xi)$ holds.
\end{thm}

In order to obtain a more manageable condition of this result  let $\hat \mu_{i,j}(\xi)$
denote  a  measure on the set ${\Theta}_{i,j}^*(\xi)$ ( $i,j=1,\ldots,\nu$) for which the function
$$
\max_{x \in \mathcal{X}} \sum_{i,j = 1}^{\nu} p_{i,j}\int_{{\Theta}_{i,j}^*(\xi)}  \big[ \eta_i(x,\overline{\theta}_{i}) - \eta_j(x,{ \theta}_{i,j}) \big]^2 \mu_{i,j}(d\theta_{i,j})
$$
attains its minimal value, and define
\begin{equation} \label{Psifct}
\Psi(x,\xi)= \sum_{i,j = 1}^{\nu} p_{i,j}  \int_{{\Theta}_{i,j}^*(\xi)}  \big[ \eta_i(x,\overline{\theta}_{i}) - \eta_j(x,{\theta}_{i,j}) \big]^2 \hat \mu_{ij}(d {\theta}_{i,j}) ~.
\end{equation}
 Note that   the function  in  \eqref{Psifct} simplifies
 to
 \begin{equation} \label{Psifctsimp}
\Psi(x,\xi)= \sum_{i,j = 1}^{\nu} p_{i,j}    \big[ \eta_i(x,\overline{\theta}_{i}) - \eta_j(x,{\widehat \theta}_{i,j}) \big]^2  ~,
\end{equation}
if  both Assumptions  \ref{assum1} and  \ref{assum2} are satisfied.
\begin{cor}\label{thm1a}
If Assumption  \ref{assum1}    is satisfied and the design  $\xi$ is not $T_\mathrm{P}$-optimal then there exists
a point  $\overline x \in \mathcal{X} $ such that
\begin{align} \nonumber
\Psi(\overline x,\xi)> T_{\mathrm{P}}(\xi).
\end{align}
\end{cor}

\subsection{Bayesian $T$-optimal designs} \label{sec2b}

 As pointed out by   \cite{detmelshp2012} locally $T$-optimal designs are rather sensitive with respect to misspecification of the unknown
 parameters  $\overline{\theta}_{i}$, and it might be appropriate to construct more robust
 designs for model discrimination.  The problem of robustness was already mentioned in \cite{atkfed1975a}
 and these authors proposed a Bayesian version of the $T$-optimality criterion which reads  in the situation of
 the  criterion  \eqref{2.4} as
 follows
 \begin{align} \label{2.5}
T_{\mathrm{P}}^{\mathrm{B}}(\xi) = \sum_{i,j=1}^{\nu} p_{i,j} \int_{\Theta_i} \inf_{\theta_{i,j} \in \Theta_j} \int_{\mathcal{X}} \Big[ \eta_i(x,\lambda_i) - \eta_j(x,\theta_{i,j}) \Big]^2 \xi(dx) \mathcal{P}_{i}(d\lambda_i).
\end{align}
Here for each $i=1, \ldots, \nu$ the measure  $\mathcal{P}_{i}$ denotes a prior distribution for the parameter $\theta_{i}$  in model
$\eta_i$, such that all integrals in \eqref{2.5} are well defined.  Throughout this paper we will call any design maximizing the criterion
\eqref{2.5} a Bayesian $T$-optimal discriminating design. For (two) polynomial regression models Bayesian $T$-optimal discriminating
designs have been explicitly determined by \cite{DetteMelasShpilev2013}, and their results indicate the
intrinsic difficulties in the construction
of optimal designs with respect to this criterion. \\
 In the following we will link the criterion \eqref{2.5} to the locally
$T$-optimality criterion \eqref{2.4} for large number of competing models.
For this purpose we note that in  nearly all  situations of practical interest an explicit evaluation of the integral in \eqref{2.5} is not possible and the
criterion has to be evaluated by numerical integration approximating  the prior distribution by a measure with finite support.
Therefore  we assume that the prior distribution $ \mathcal{P}_{i}$ in the criterion is given  by a discrete measure with
masses $\tau_{i1}, \ldots \tau_{i\ell_i}$ at the points $\lambda_{i1},\ldots ,\lambda_{i\ell_i}$. The criterion in \eqref{2.5} can
then  be  rewritten as
 \begin{align} \label{2.6}
T_{\mathrm{P}}^{\mathrm{B}}(\xi) = \sum_{i,j=1}^{\nu} \sum_{k=1}^{\ell_i} p_{i,j}  \tau_{ik} \inf_{\theta_{i,j} \in \Theta_j} \int_{\mathcal{X}} \big[ \eta_i(x,\lambda_{ik}) - \eta_j(x,\theta_{i,j}) \big]^2 \xi(dx) ,
\end{align}
which is a locally $T$-optimality criterion of the from \eqref{2.4}. The only difference between the criterion obtained form the Bayesian approach
and  \eqref{2.4} consists in the fact that the criterion \eqref{2.6} involves substantially more comparisons of the functions $\eta_i$ and $\eta_j$.
For example, if this approach  is used for a Bayesian version of the criterion \eqref{2.2} we obtain
\begin{align} \label{2.7}
T_{12}^B (\xi) = \sum_{k=1}^{\ell}   \tau_{k}
\inf_{\theta_{2} \in \Theta_2} \int_{\mathcal{X}} \big[ \eta_1(x, \lambda_k) - \eta_2(x,\theta_{2}) \big]^2 \xi(dx).
\end{align}
This is the locally $T$-optimality criterion \eqref{2.4} with $\nu=\ell +1,\ p_{i,\ell +1}=\tau_i \ (i=1,\dots,\ell)$ and $p_{i,j}=0$ otherwise.
Thus, instead of making only one comparison as required for the locally $T$-optimality criterion, the Bayesian approach
(with a discrete approximation of the prior) yields a criterion with $\ell $ comparisons, where $\ell$ denotes the number of support points
used for the approximation of the prior distribution. Moreover, for each support point of the prior distribution in the criterion \eqref{2.6} (or \eqref{2.7}) the infimum has to be calculated numerically, which is computationally expensive. Consequently, the computation of Bayesian $T$-optimal discriminating
design problems is particularly challenging. In the following sections we provide an efficient solution of  this problem.

\section{Calculating locally $T$-optimal designs }\label{sec3}

\cite{BraessDette2013} proposed an algorithm for the numerical construction of
locally $T$-optimal designs,  which is based on vector-valued  Chebyshev approximation. This algorithm is quite difficult both
in terms of description and  implementation. Moreover, it requires   substantial computational resources and is therefore only able
to deal  with a small number of comparisons in the $T$-optimality criterion. The purpose of this section is to develop
a more efficient method which is able to deal with a large number of comparisons in the the criterion and  avoids the  drawbacks of the procedures
in \cite{atkfed1975a} and \cite{BraessDette2013}. As pointed out in Section \ref{sec2b}
methods solving this problem are required for the calculation of Bayesian $T$-optimal discriminating designs.

  Recall the definition of the function $\Psi$ in \eqref{Psifct} and  note that under Assumption  \ref{assum1}   it follows from
Corollary  \ref{thm1a}
that there exists  a point  $\overline{x} \in \mathcal{X}$, such that the inequality
$$
\Psi(\overline{x},\xi) > T_{\mathrm{P}}(\xi)
$$
holds, whenever  $\xi$ is {\bf not}  a locally $T$-optimal discriminating design. The algorithm of
\cite{atkfed1975a} uses this property to construct a sequence of designs which converges to the locally  $T$-optimal discriminating design.
 For further reference it is stated here.

\begin{algorithm}[\cite{atkfed1975a}]{\ }
\label{algorithm:AtkinsonFedorov}
Let $ \xi_0$ denote a given (starting) design and let $( \alpha_s)_{s=0}^{\infty}$   be  a sequence of positive numbers, such that
 $\lim_{s \to \infty} \, \alpha_s = 0, \; \sum_{s = 0}^{\infty} \alpha_s = \infty, \; \sum_{s = 0}^{\infty} \alpha_s^2 < \infty. $
For $s=0,1,\ldots $ define
$$\xi_{s+1} = ( 1 - \alpha_s ) \xi_s + \alpha_s \xi(x_{s+1}),$$
 where
$ x_{s+1} = \argmax_{x \in \mathcal{X}} \Psi(x,\xi_s).$
\end{algorithm}
\noindent
It can be shown that this algorithm converges in the sense that
$\lim_{s \to \infty} T_{\mathrm{P}}(\xi_s) =  T_{\mathrm{P}}(\xi^*)$, where $\xi^*$ denotes a locally $T$-optimal discriminating design.
However,  a  major problem of Algorithm \ref{algorithm:AtkinsonFedorov} is that it yields a sequence of designs
with an  increasing number  of  support points. As a consequence  the resulting design (after applying some
stopping criterion)  is concentrated on a large set of points.
Even  if this  problem can be solved  by clustering or  by determining the extrema of
the final function $\Psi(x,\xi_s)$, it is much more difficult to deal with the accumulation of  support points  during the iteration.
Moreover, \cite{BraessDette2013} demonstrated that in many cases
the iteration process may take several hundred iterations for obtaining a
locally $T$- optimal discriminating design with a required precision,  resulting in a
 high computational complexity for the recalculation of the optimum values
\begin{align}\label{thetatilde}
 \widehat{\theta}_{i,j} \in \arginf_{\theta_{i,j} \in \Theta_j} \int_{\mathcal{X}} \big[ \eta_i(x,\overline{\theta}_i) - \eta_j(x, \theta_{i,j}) \big]^2 \xi(dx)
\end{align}
in the  optimality criterion \eqref{2.4}. These authors also showed that   Algorithm \ref{algorithm:AtkinsonFedorov}
may not find the optimal design if there are too many model comparisons involved in the
$T$-optimality criterion \eqref{2.4}. \\
Therefore, we propose the following  alternative basic procedure for the calculation of locally $T$-optimal discriminating designs
as an alternative to Algorithm~\ref{algorithm:AtkinsonFedorov}. Roughly speaking, it consists of two steps treating the maximization with respect to support points (Step 1) and weights (Step 2)
separately, where two methods implementing  the second step will be given below [see Section \ref{sec31} and \ref{sec32} for details].
\begin{algorithm}{\ }
\label{algorithm:new}{\rm
Let ${\xi_0}$ denote a starting design such that  $T_{\mathrm{P}}({\xi_0})>0$ and define recursively
a sequence of designs $({\xi_s})_{s=0,1,\ldots } $  as follows:
\begin{itemize}
\item[$(1)$]  Let $\mathcal{S}_{[s]}$ denote the support of the design ${\xi_s}$. Determine
the set $\mathcal{E}_{[s]}$ of all local maxima of  the function $\Psi(x,{\xi_s})$ on the design space
 $\mathcal{X}$ and define
$\mathcal{S}_{[s+1]} =  \mathcal{S}_{[s]}  \cup \mathcal{E}_{[s]}$.
\item[$(2)$] We define $ \xi = \{\mathcal{S}_{[s+1]},\omega\} $ as the design supported at $ \mathcal{S}_{[s+1]}$ (with a vector $w$ of weights) and determine the locally $T_P$-optimal
design in the class of all  designs supported at $ \mathcal{S}_{[s+1]}$, that is we determine
the vector $\omega_{[s+1]}$ maximizing the function
\begin{align*}
g(\omega) = T_{\mathrm{P}}( \{\mathcal{S}_{[s+1]},\omega\}) = \sum_{i,j=1}^{\nu} p_{i,j} \inf_{\theta_{i,j} \in \Theta_j}
\sum_{x \in  \mathcal{S}_{[s+1]}} \big[ \eta_i(x,\overline{\theta}_{i}) - \eta_j(x,\theta_{i,j}) \big]^2 w_x
\end{align*}
(here $w_x$ denotes the weights at the point $x \in \mathcal{S}_{s+1}$).
 All points in ${\mathcal{S}}_{[s+1]}$ with vanishing components in the vector of weights   $\omega_{[s+1]}$  will be
 be removed and the new set of support points will also be denoted by ${\mathcal{S}}_{[s+1]}$.
 Finally the design ${\xi}_{s+1}$ is defined as the design with the set of support points  ${\mathcal{S}}_{[s+1]}$ and the corresponding nonzero weights.
\end{itemize}
}
\end{algorithm}

\noindent
\begin{thm}\label{thm2}
Let Assumption \ref{assum1} be satisfied and
let $({\xi_s})_{s=0,1,\ldots } $ denote the sequence of designs obtained by  Algorithm \ref{algorithm:new},
 then $$\lim_{s \to \infty} T_{\mathrm{P}}(\xi_{s+1}) =  T_{\mathrm{P}}(\xi^*),$$
  where $\xi^*$ denotes a   locally $T$-optimal discriminating  design.
\end{thm}

A proof of Theorem \ref{thm2} is deferred to  Section \ref{sec6}.
Note that the algorithm adds all local maxima of the function $\Psi(x,{\xi_s})$ as possible support points of the design in the next  iteration.
Consequently, in  the current form Algorithm~\ref{algorithm:new} also  accumulates  too many
support points. To avoid this problem, it is suggested to remove at each step those
points from the support, whenever their weight is smaller than $m^{0.25}$, where $m$ denote the working
precision of the software used in the implementation  (which is $2.2\times 10^{-16} $ for $R$).
 Note also that this  refinement does not affect  the convergence of the algorithm  from a practical point of view.
 A more important question is the implementation
 of the second step of the procedure, that is   the maximization of function $g(\omega)$. Before we
 discuss two computationally
 efficient procedures for this purpose in the following sections, we
state an important property
of the function $\Psi(x,\xi_{s+1})$ obtained  in each iteration.
\begin{lemma}\label{lem2}
At the end of each iteration of Algorithm \ref{algorithm:new}  the function
$
\Psi(x,\xi_{s+1})
$
attains one and the same value for all support points of the design $ \xi_{s+1}$.
\end{lemma}

\subsection{Quadratic programming} \label{sec31}
Let  $\mathcal{S}_{[s+1]} = \{x_1,\ldots , x_n\}$ denote the set
obtained in the first step of Algorithm \ref{algorithm:new} and define $\xi$ as a design supported at   $\mathcal{S}_{[s+1]} $ with
corresponding weights $\omega_1,\ldots , \omega_n$ (which have to be determined in Step 2 of the algorithm by
maximizing the function
\begin{align*}
g(\omega) = \sum_{i,j = 1}^{\nu} p_{i,j}  \sum_{k = 1}^{n} {\omega}_k \bigl[ \eta_i(x_k,\overline{\theta}_{i}) - \eta_j(x_k, \widehat{\theta}_{i,j}  ) \bigr]^2,
\end{align*}
where $ \widehat{\theta}_{i,j} = \widehat{\theta}_{i,j}  (\omega) $  is defined in \eqref{thetatilde}. For this purpose we   suggest to linearize the functions  $\eta_j(x_k,\theta_{i,j})$ in the neighborhood of point $\widehat{\theta}_{i,j}$. More precisely, we consider the function
\begin{align*}
\overline g(\omega) &= \sum_{i,j = 1}^{\nu} p_{i,j} \min_{\alpha_{i,j} \in \R^{d_j} } \sum_{k = 1}^{n} {\omega}_k \Big[ \eta_i(x_k,\overline{\theta}_{i}) - \eta_j(x_k,\widehat{\theta}_{i,j})  - \alpha_{i,j} ^T
 \frac{\partial \eta_j(x_k,{\theta}_{i,j})}{\partial {\theta}_{i,j}}\Big|_{{\theta}_{i,j} = \widehat{\theta}_{i,j}} \Big]^2. \\
&= \sum_{i,j = 1}^{\nu} p_{i,j} \min_{\alpha_{i,j} \in \R^{d_j}  } \left[ \alpha_{i,j}^{\mathrm{T}} \mathbf{J}_{i,j}^{\mathrm{T}} \mathbf{\Omega} \mathbf{J}_{i,j} \alpha_{i,j}
- 2 \omega^{\mathrm{T}} \mathbf{R}_{i,j} \alpha_{i,j}  + b_{i,j}^{\mathrm{T}} \omega \right],
\end{align*}
where  $d_j$ is the dimension of the parameter space $\Theta_j$, $ \mathbf{\Omega} = \mathrm{diag}({\omega}_1,\dots,{\omega}_{n})$ and the matrices $\mathbf{J}_{i,j}  \in \R^{n \times d_j} $, $\mathbf{R}_{i,j}  \in \R^{n \times d_j} $ and the vectors $b_{i,j}
 \in \R^{n }$ are defined by
\begin{align*}
&\mathbf{J}_{i,j} = \Bigl( \frac{\partial \eta_j(x_r,\theta_{i,j})}{\partial \theta_{i,j}}\Big|_{\theta_{i,j} = \widehat{\theta}_{i,j}} \Bigl)_{ r = 1,\dots,n},    \\
&\mathbf{R}_{i,j} = \Bigl( [\eta_i(x_r,\overline{\theta}_{i}) - \eta_j(x_r,\widehat{\theta}_{i,j}) ] \frac{\partial \eta_j(x_r,\theta_{i,j})}{\partial \theta_{i,j}}\Big|_{\theta_{i,j}
= \widehat{\theta}_{i,j}} \Bigl)_{ r = 1,\dots,n },\\
&b_{i,j}  = \Bigl(  [\eta_i(x_r,\overline{\theta}_{i}) - \eta_j(x_r,\widehat{\theta}_{i,j})]^2 \Bigl)_{r = 1,\dots,n},
\end{align*}
respectively. Obviously the minimum  with respect to
$\alpha_{i,j}$ is achieved by $\alpha_{i,j} = \left(\mathbf{J}_{i,j}^{\mathrm{T}} \mathbf{\Omega} \mathbf{J}_{i,j}\right)^{-1} \mathbf{R}_{i,j}^{\mathrm{T}} \omega$
which gives
\begin{align*}
\overline g(\omega) = -\omega^{\mathrm{T}} \mathbf{Q}  (\omega )  \; \omega + b^{\mathrm{T}} \omega,
\end{align*}
where
\begin{align*}
&\mathbf{Q}  (\omega )  = \sum_{i,j = 1}^{\nu} p_{i,j} \mathbf{R}_{i,j} \left( \mathbf{J}_{i,j}^{\mathrm{T}} \mathbf{\Omega} \mathbf{J}_{i,j} \right)^{-1} \mathbf{R}_{i,j}^{\mathrm{T}}.
\end{align*}
The matrix $\mathbf{Q}  (\omega ) $ depends on $\omega$, but if we ignore this dependence and take the matrix $\mathbf{\Omega} = \mathrm{diag}(\overline{\omega}_1,\dots,\overline{\omega}_{n})$ as fixed, then we end up with a quadratic programming problem, that is
\begin{align}
&\phi(\omega,\overline{\omega}) = -\omega^{\mathrm{T}} \mathbf{Q(\overline{\omega})} \; \omega + b^{\mathrm{T}} \omega \rightarrow \max_{\omega},  \label{iteration} \\
&\sum_{k = 1}^{n} \omega_k = 1; \; \omega_k \geq 0, \; k = 1,\dots,n. \nonumber
\end{align}
This problem is   solved iteratively until convergence, substituting each time the solution obtained in the previous iteration instead of $\overline{\omega}$. We
note that a similar idea has also been  proposed by \cite{BraessDette2013}.

\begin{rem} {\rm
In the practical implementation of the procedure it is recommended to perform only a few iterations of this step such  that
an improvement in the
difference between the value of the criterion of the starting design in Step 2 and the design obtained in the iteration of \eqref{iteration} is
observed. This will speed up the convergence of the procedure substantially.
In this case  equality of  the function $\Psi$ at the  support points of the calculated design  (as stated in Lemma~\ref{lem2}) is only achieved approximately. \\
Formally, the convergence of the algorithm is
only proved if the   iteration \eqref{iteration} is performed until convergence. However, in all examples considered so far, we  observed
convergence of the procedure, even if only a few iterations of  \eqref{iteration} are used.
 In our $R$ program the user can specify the number of iterations used in this part of the algorithm. Thus,
if any problem regarding convergence is observed, the number of iterations should be increased (of course at a cost speed
of the algorithm).
}
\end{rem}

\subsection{A gradient method} \label{sec32}

A further  option for the second step in  Algorithm~\ref{algorithm:new} is a specialized gradient method, which is used for the
function
\begin{align}\label{gfct}
g(\omega) =  \sum_{i,j = 1}^{\nu} p_{i,j}
\sum_{k = 1}^{n} {\omega}_k \big[ \eta_i(x_k,\overline{\theta}_{i}) - \eta_j(x_k,\widehat{\theta}_{i,j}  ) \big]^2
\end{align}
where $\widehat{\theta}_{i,j}  = \widehat{\theta}_{i,j} (\omega) $ is defined in \eqref{thetatilde}.
For it s description we define  the functions
\begin{align*}
v_k(\omega) = \sum_{i,j = 1}^{\nu} p_{i,j} \bigl[  \eta_i(x_k,\theta_{i})  - \eta_j(x_k,\widehat{\theta}_{i,j}(\omega))  \bigr]^2, \;  k = 1,\dots,n,
\end{align*}
and iteratively calculate a sequence of vectors  $(\omega_{(\gamma)})_{\gamma =0,1,\ldots}$.  At the beginning we choose
$\omega_{(0)} = \overline{\omega}$ (for example equal weights). If
 ${\omega}_{(\gamma )} =({\omega}_{(\gamma ),1}, \ldots,{\omega}_{(\gamma ),n})$
 is given, we proceed for $\gamma=0,1,\ldots $ as follows. We determine indices
 $\overline{k} $ and  $ \underline{k}$ corresponding to
$ \max_{1 \leq k \leq n} v_k(\omega_{(\gamma)})$ and
$\min_{1 \leq k \leq n} v_k(\omega_{(\gamma)})$, respectively, and define
\begin{align} \label{argmax}
\alpha^* = \arg \max_{0 \leq \alpha \leq \omega_{(\gamma),\underline{k}}} g(\overline{\omega}_{(\gamma)}(\alpha)),
\end{align}
where the vector  $\overline{\omega}_{(\gamma )} (\alpha)
=(\overline{\omega}_{(\gamma ),1}(\alpha), \ldots,\overline{\omega}_{(\gamma ),n}(\alpha))$
is  given by
$$
\overline{\omega}_{(\gamma),i} (\alpha)= \left\{ \begin{array}{ll}
\omega_{(\gamma),i}  + \alpha & \mbox{ if } i=\overline{k}\\
\omega_{(\gamma),i}  -  \alpha & \mbox{ if }  i= \underline{k}\\
\omega_{(\gamma),i}   & \mbox{ else } \\
\end{array}
\right.
$$
The vector ${\omega}_{(\gamma+1)} $ of the next iteration is then defined by  $
{\omega}_{(\gamma+1)}  = \overline{\omega}_{(\gamma)} (\alpha^*).$
The following theorem shows that the generated sequence of vectors converges to a maximizer of the function $g$ in \eqref{gfct} and is proved in   the Appendix.
\begin{thm}
\label{conv_theorem}
The sequence $(\omega_{(\gamma)})_{\gamma \in \mathbb{N}}$ converges to a vector
$\omega^* \in \argmax g(\omega)$.
\end{thm}

\begin{rem}{\rm
It is worthwhile to  mention  that the one dimensional optimization problem \eqref{argmax} is  computationally  rather expensive.
In the   implementation we use a  linearization of the optimization problem, which is obtained  in a
similar way  a described in Section \ref{sec31}.
}
\end{rem}

\section{Implementation and numerical examples }\label{sec5}

We have implemented the procedure for the calculation of the locally $T$-optimal discriminating design
 in $R$, where the user has to specify the weights $p_{i,j}$
 and the  corresponding  preliminary information regarding the parameters $\overline{\theta}_i$.
 To be precise, we  call
\begin{align*}
\mathrm{P} = \left[
\begin{array}{ccccc}
p_{1,1} & p_{1,2} & \dots & p_{1,\nu-1} & p_{1,\nu} \\
\vdots & \vdots & \vdots & \vdots & \vdots\\
p_{\nu,1} & p_{\nu,2} & \dots & p_{\nu,\nu-1} & p_{\nu,\nu} \\
\end{array}
\right]
\end{align*}
the comparison table for the  locally $T$-optimal discriminating design problem under consideration.
This table has to be specified by the experimenter. Because the Bayesian $T$-optimal design problem
with a discrete prior  can be reduced to a locally  $T$-optimal   one with a
large number of model comparisons,
we now describe the corresponding table for the Bayesian $T$-optimality criterion.
For illustration purposes we consider
 the case $\nu=2$. The Bayesian $T$-optimality  criterion is given in \eqref{2.7}, where the prior for
 the parameter $\theta_1$ puts masses
 $\tau_{1}, \ldots \tau_{\ell}$ at the points $\lambda_{1},\ldots ,\lambda_{\ell}$. This criterion can  be rewritten as
 a local $T$-optimality criterion of the form \eqref{2.4}, i.e.
\begin{align}
T_{\mathrm{P}}(\xi) = \sum_{i,j=1}^{\ell+1}  {p}_{i,j} \inf_{\theta_{i,j} \in \Theta_j} \int_{\mathcal{X}} \Big[ \eta_i(x, {\theta}_{i}) - \eta_j(x,\theta_{i,j}) \Big]^2 \xi(dx),
\end{align}
where comparison table is given by
\begin{align}
{\mathrm{P}} = ({p}_{i,j})_{i,j=1,\ldots ,\ell+1}~=~  \left[ \begin{array}{ccccc}
 0 &  0 & \dots &  0 &  \tau_1 \\
0 &  0 & \dots & 0 &  \tau_2 \\
\vdots & \vdots & \vdots & \vdots & \vdots\\
0 & 0 & \dots &  0 &  \tau_\ell \\
0 & 0 & \dots & 0 & 0 \\
\end{array} \right] \in \R^{\ell +1 \times \ell +1 } ,
\end{align}
$\eta_i(x,\overline{\theta}_{i}) = \eta_1(x,\lambda_i ), \; i = 1,\dots, \ell$ and $\eta_{\ell + 1}(x,\theta_{i,j}) = \eta_2(x,\theta_{i,\ell + 1})$. The extension of this approach to
more than two models is easy and left to the reader. We now illustrate the new method in two examples  calculating Bayesian $T$-optimal discriminating designs.
We have implemented both procedures described in Section \ref{sec31} and \ref{sec32} and the results were similar. For this reason we only represent
 the Bayesian $T$-optimal discriminating designs  calculated by Algorithm \ref{algorithm:new}, where  the quadratic programming method was used in
  Step 2  [see Section \ref{sec31}
 for details].

\subsection{Bayesian $T$-optimal discriminating designs for  exponential models} \label{sec51}

Consider the problem of discriminating between the  two regression models
\begin{align}
&\eta_1(x,\theta_1) = \theta_{1,1} - \theta_{1,2} \exp(-\theta_{1,3} x^{\theta_{1,4}}),  \label{example1}\\
&\eta_2(x,\theta_2) = \theta_{2,1} - \theta_{2,2} \exp(-\theta_{2,3} x),\nonumber
\end{align}
where the design  space is given by the interval $[0,10]$.
Exponential models of the form \eqref{example1} are widely used in  applications. For
example, the    model $\eta_2$  is frequently fitted  in agricultural sciences, where it is called
 Mitscherlichs growth law and used for describing the relation between
the yield of a crop and the amount of fertilizer.
In fisheries research this model is  called  Bertalanffy growth curve  and used
for the description of the length of a fish in dependence of its age
[see \cite{ratkowsky1990}].
Optimal designs for exponential regression models  have been
determined  by  \cite{hancha2003}  among others. In the following we will demonstrate the performance of the new
algorithm in calculating Bayesian $T$-optimal discriminating designs for the two exponential models. Note that it make only sense
to consider the Bayesian version of $T_{12}$, because the model $\eta_2$ is obtained as a special case of $\eta_1$ for $\theta_{1,4}=1$.
  It is easy to see that the locally $T$-optimal discriminating designs do not depend on the linear parameters of  $\eta_1$ and
  we have chosen $\bar \theta_{1,1}=2$ and $\bar \theta_{2,2}=1 $ for these parameters. For the parameters $\bar \theta_{1,3}$ and $\bar \theta_{1,4}$ we considered independent
  prior distributions supported at the points
  \begin{equation} \label{sup1}
  \mu_j + \frac {\sigma(i-3)}{2} \qquad \qquad i=1,\dots,5 \, ; \quad j=3,4 \, ,
  \end{equation}
  where $\mu_3 = 0.8, \ \mu_4 = 1.5$ and different values of the variance $\sigma^2$ are investigated. The corresponding weights at these points are
  proportional  (in both cases) to
  \begin{equation} \label{wei1}
  \frac {1}{\sqrt{2 \pi \sigma^2}} \exp \Bigl( - \frac {(i-3)^2}{8} \Bigr); \qquad \qquad i=1,\dots, 5\, .
  \end{equation}
 We note that this yields $25$ terms in the Bayesian optimality criterion \eqref{2.7}.
   Bayesian $T$-optimal discriminating designs are depicted in Table \ref{tab1}  for various values of $\sigma^2$, where an equidistant design at $11$ points $0,1,\dots,10$ was used as starting design.
   \begin{table}[h]
\begin{center}
\begin{tabular}{||c||c||c||c||}
\hline
\hline
$\sigma^2$ & Optimal design & $\sigma^2$ & Optimal design\\
\hline
\hline
0.0 & $\begin{matrix}
0.000 & 0.441 & 1.952 & 10.000 \\
0.209 & 0.385 & 0.291 & 0.115
\end{matrix}$ & 0.285 & $\begin{matrix}
0.000 & 0.453 & 1.758 & 10.000 \\
0.207 & 0.396 & 0.292 & 0.105
\end{matrix}$\\
\hline
0.1 & $\begin{matrix}
0.000 & 0.452 & 1.877 & 10.000 \\
0.209 & 0.391 & 0.290 & 0.110
\end{matrix}$ & 0.3 & $\begin{matrix}
0.000 & 0.452 & 1.747 & 4.951 & 10.000 \\
0.207 & 0.396 & 0.292 & 0.003 & 0.102
\end{matrix}$\\
\hline
0.2 & $\begin{matrix}
0.000 & 0.455 & 1.811 & 10.000 \\
0.208 & 0.394 & 0.291 & 0.107
\end{matrix}$ & 0.4 & $\begin{matrix}
0.000 & 0.446 & 1.651 &  4.699 & 10.000 \\
0.200 & 0.384 & 0.290 & 0.060 & 0.066
\end{matrix}$\\
\hline
\hline
\end{tabular}
\end{center}
 \caption{\it \label{tab1} Bayesian $T$-optimal discriminating designs for the two exponential models in \eqref{example1}.  The support points and weights of
 the  independent prior distributions for the parameters $\overline{\theta}_{1,3}$  and $\overline{\theta}_{1,4}$   are given by \eqref{sup1} and \eqref{wei1}, respectively.}
\end{table}
   
  A typical determination of the
  optimal design takes  between
$0.03$  seconds  (in the case  $\sigma^2=0$)   and  $1.4$  seconds  (in the case  $\sigma^2=0.4$)    CPU  time on a standard PC (with an intel core i7-4790K processor).
The algorithm using the procedure described in Section \ref{sec32}
in  step 2 requires  between
$0.11$  seconds  (in the case  $\sigma^2=0$)   and  $11.6$  seconds  (in the case  $\sigma^2=0.4$)    CPU  time. We observe that for small values of $\sigma^2$ the optimal designs are supported at $4$ points, while for $\sigma^2 \geq 0.285$ the Bayesian $T$-optimal discriminating design is supported at $5$ points.
 The corresponding function $\Psi$ from the equivalence Theorem \ref{thm1}.
  is shown in Figure \ref{fig1}.
\begin{figure}[h]
\centering
  {\includegraphics[width=45mm]{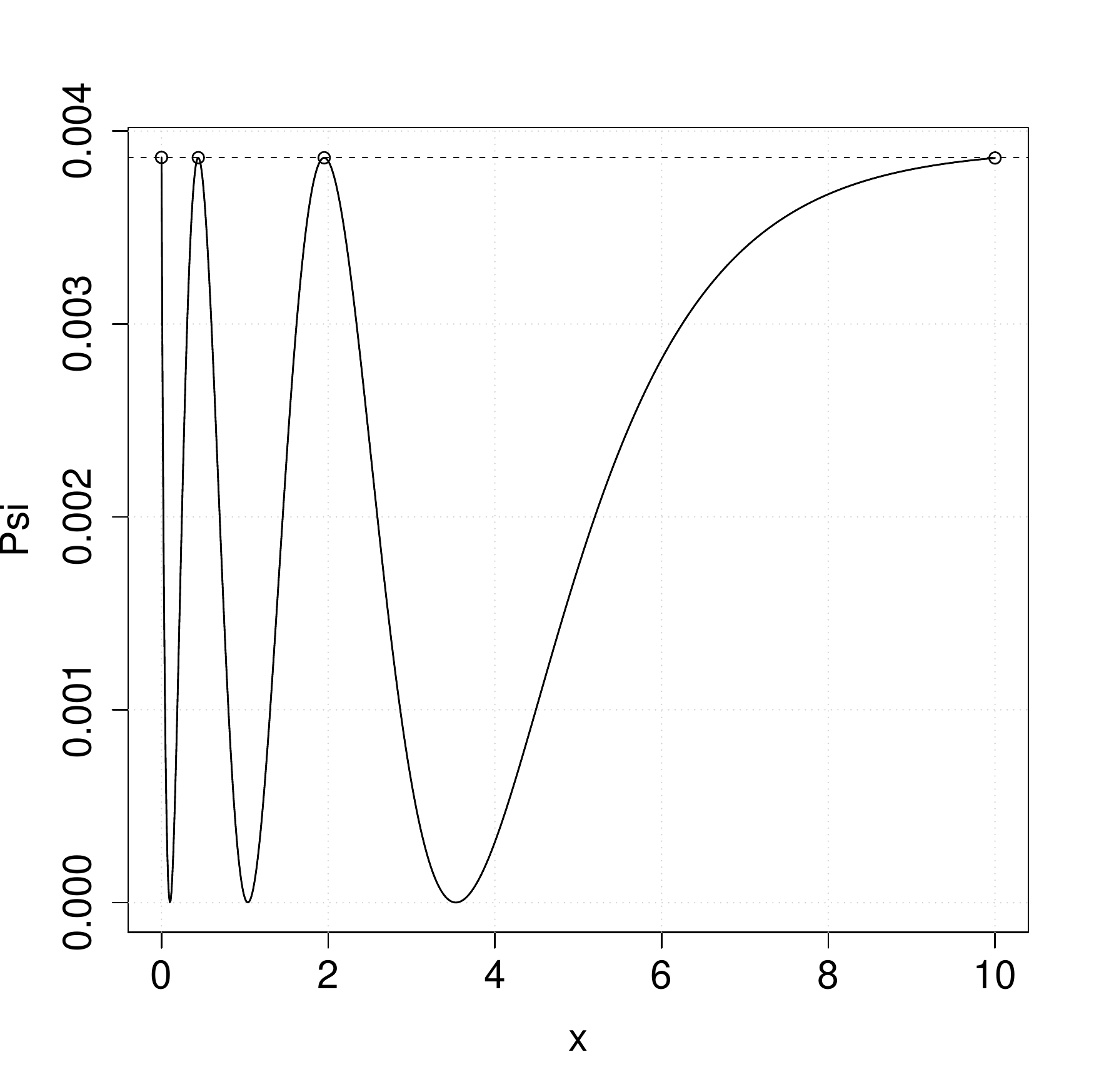}}
  {\includegraphics[width=45mm]{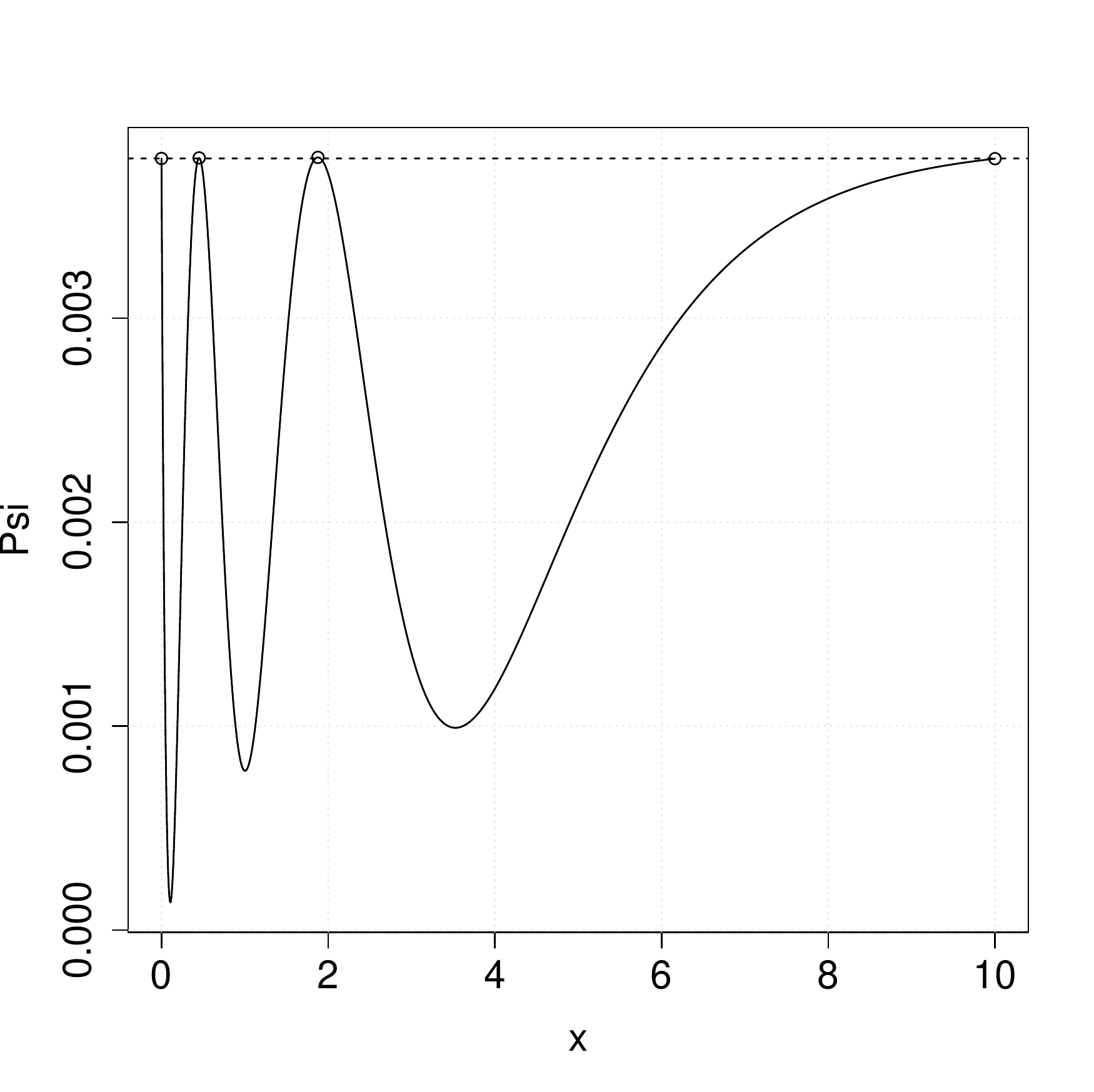}}
  {\includegraphics[width=45mm]{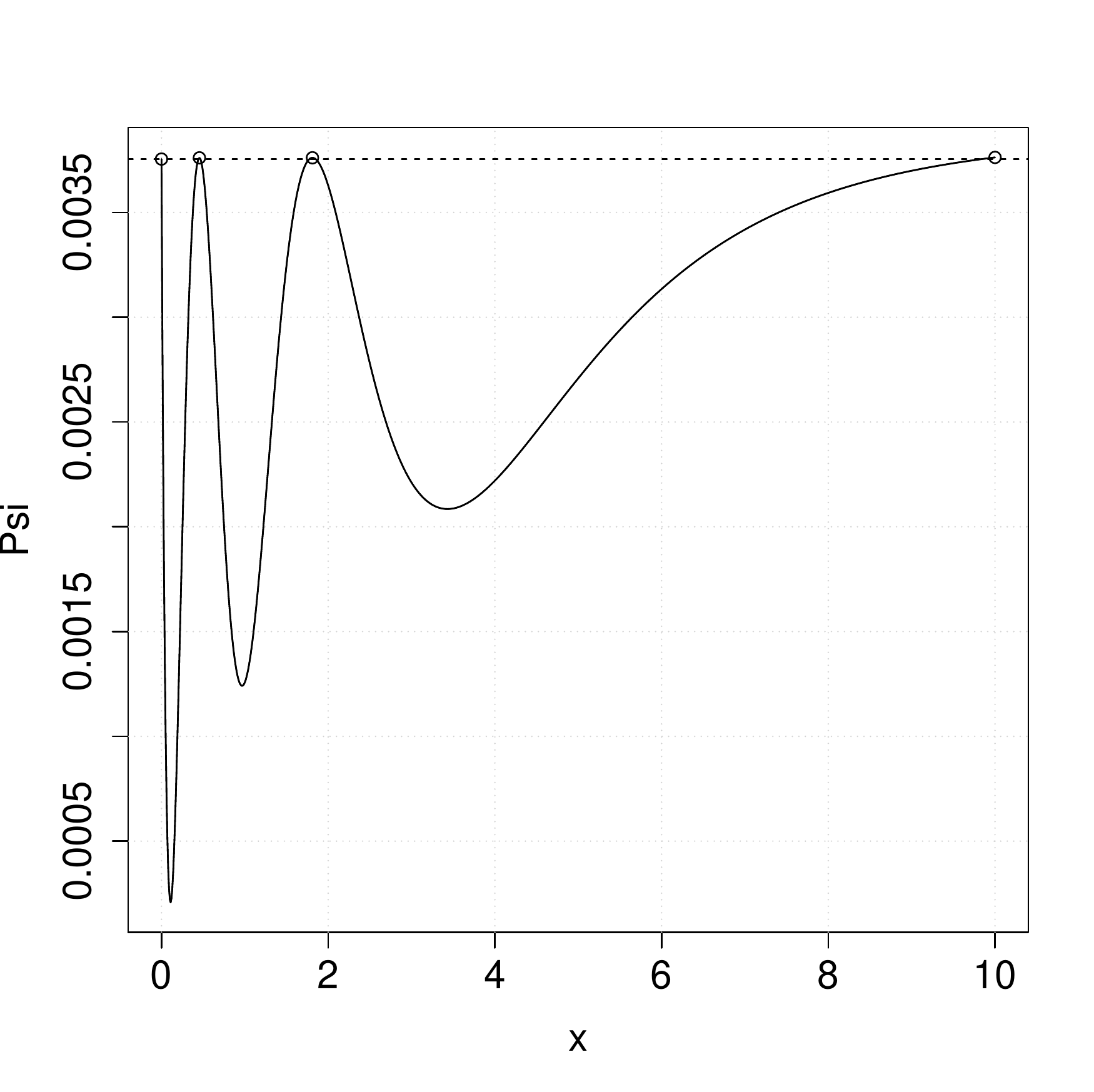}}   \\
$\sigma^2 = 0$~~~~~~~~~~~~~~~~~~~~~~~~~$\sigma^2 = 0.1$ ~~~~~~~~~~~~~~~~~~~~ $\sigma^2 = 0.2$ \\
  {\includegraphics[width=45mm]{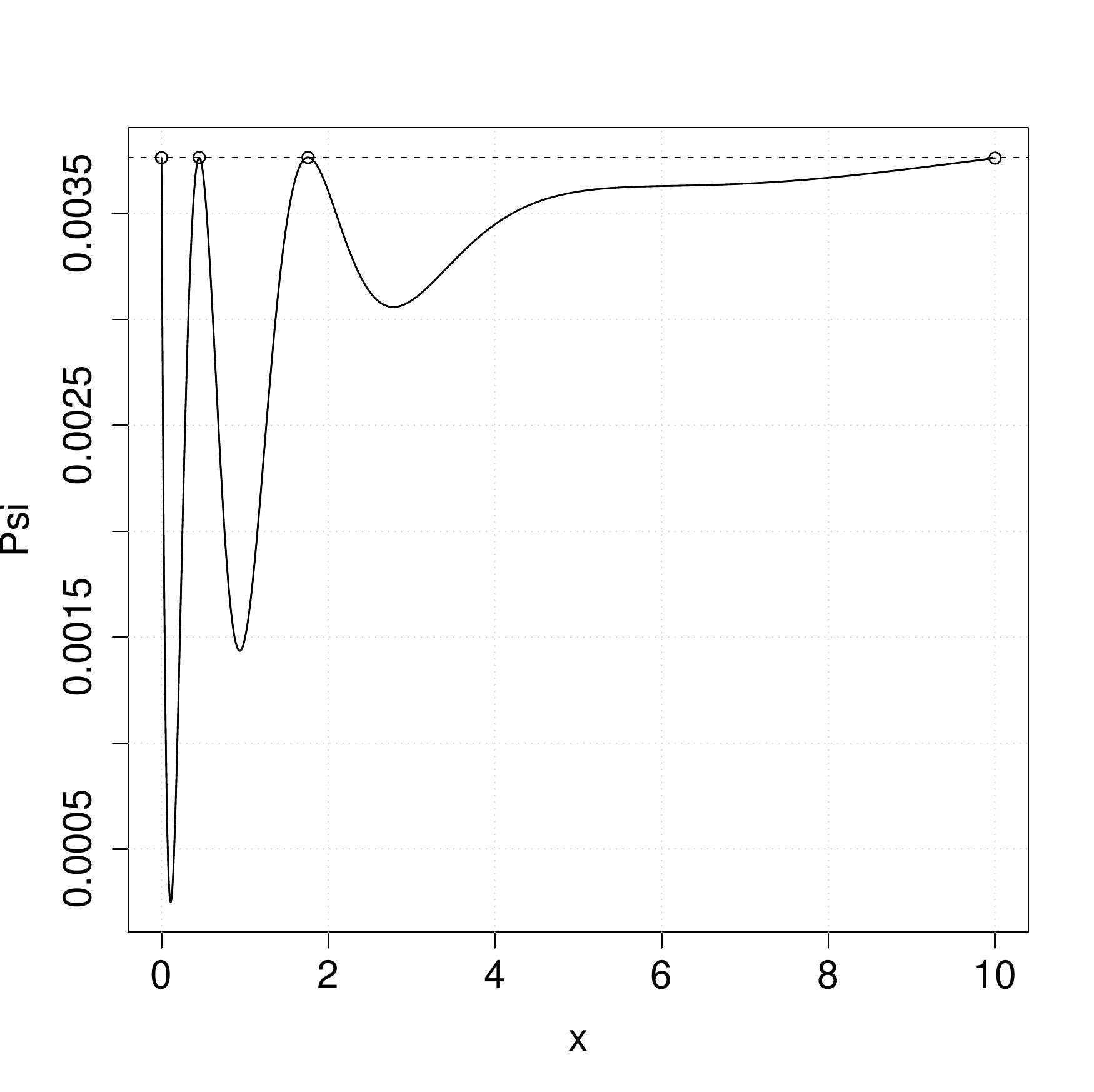}}
  {\includegraphics[width=45mm]{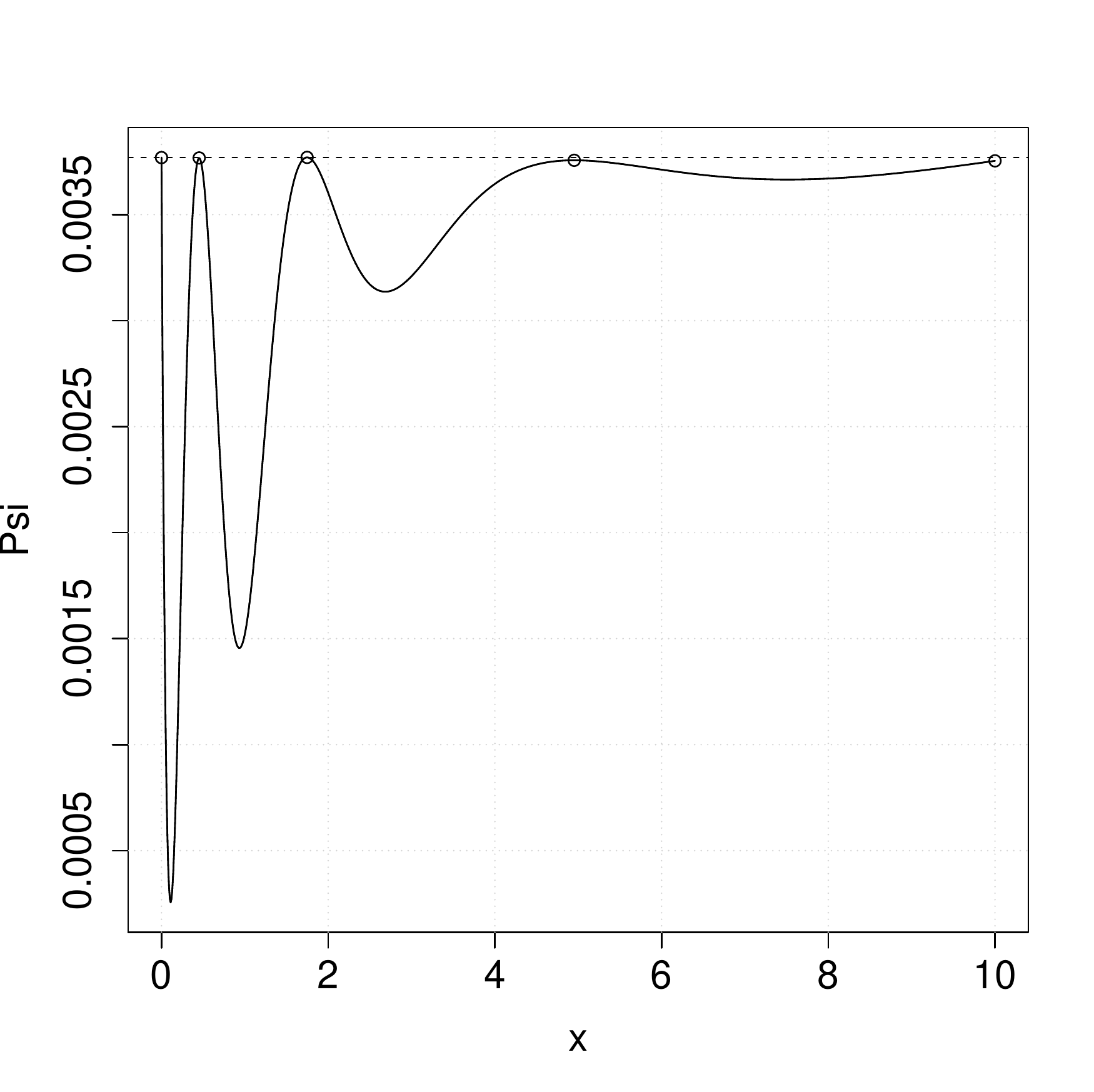}}
  {\includegraphics[width=45mm]{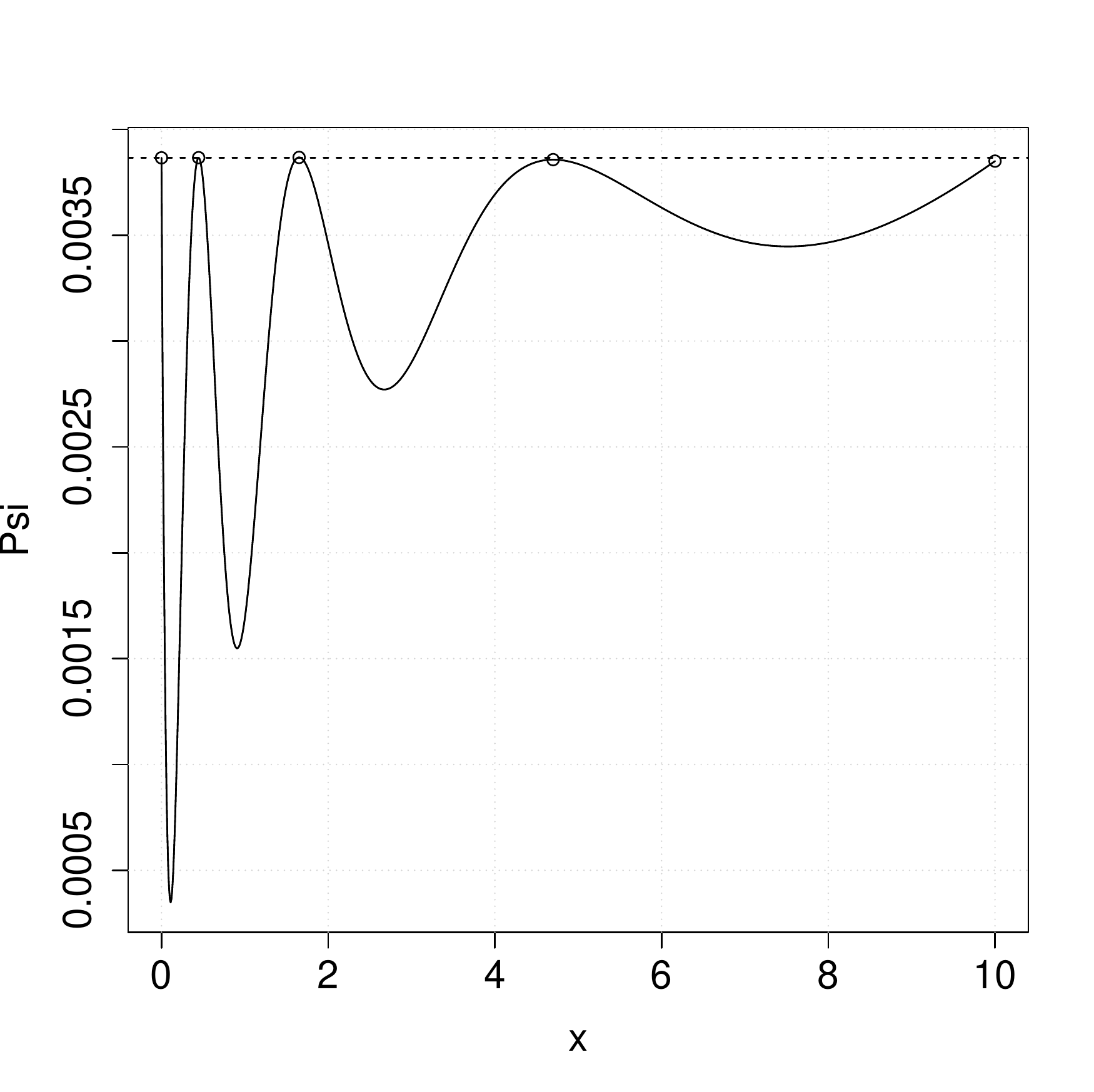}}
$\sigma^2 = 0.285$~~~~~~~~~~~~~~~~~~~~~~~~$\sigma^2 = 0.3$ ~~~~~~~~~~~~~~~~~~~ $\sigma^2 = 0.4$
 \caption{\it \label{fig1} The function on the left hand side of inequality
\eqref{equiv} in the equivalence Theorem \ref{thm1} for the numerically calculated Bayesian $T$-optimal discriminating designs.
The  competing regression models are given in \eqref{example1}.}
\end{figure}

\subsection{Bayesian $T$-optimal discrimination designs for dose finding studies} \label{sec52}
Non-linear regression models have  also numerous applications in dose response studies, where they
are used to describe the dose response relationship. In these and similar situations  the first step of the data
analysis consists in the identification of an appropriate model, and the design of experiment should take this
task into account.  For example, for modeling  the dose response relationship of a Phase II clinical trial
\cite{pinbrebra2006} proposed  the following plausible models
\begin{align}
&\eta_1(x, \theta_1) = \theta_{1,1} + \theta_{1,2} x ;  \nonumber \\
&\eta_2(x, \theta_2) = \theta_{2,1} + \theta_{2,2} x (\theta_{2,3} - x) ; \label{example2}  \\
&\eta_3(x, \theta_3) = \theta_{3,1} + \theta_{3,2} x / (\theta_{3,3} + x) ; \nonumber\\
&\eta_4(x, \theta_4) = \theta_{4,1} + \theta_{4,2} / (1 + \exp(\theta_{4,3} - x) / \theta_{4,4}) ; \nonumber
\end{align}
where the designs space (dose range) is given by the interval
$\mathcal{X}=[0,500]$. In this reference  some prior information  regarding the parameters for the  models is also provided, that is
\begin{align*}
\overline{\theta}_1= (60, 0.56), \; \overline{\theta}_2= (60, 7/2250, 600), \; \overline{\theta}_3 = (60, 294, 25), \; \ \overline{\theta}_4=
(49.62, 290.51, 150, 45.51).
\end{align*}
Locally optimal discrimination designs for the models in \eqref{example2} have been determined by  \cite{BraessDette2013}  in the case
$p_{i,j}=1/6$, $(1\leq j<i \leq 4 )$, which means that the resulting local $T$-optimality criterion \eqref{2.4} consists of $6$  model comparisons.
 \\
We begin with an illustration of the new  methodology developed in Section \ref{sec3} calculating again the locally $T$-optimal discriminating design
for this scenario.
The proposed algorithm  needs  only  four iterations for the calculation of a design, say $\xi_{4}$, which has at least
efficiency
\begin{align*}
\mathrm{Eff}_{T_{\mathrm{P}}}( \xi_{ 4 } ) = \frac{T_{\mathrm{P}}(\tilde \xi_{4} )}{\sup_{\zeta} T_{\mathrm{P}}(\zeta)} ~ \geq ~0.999.
\end{align*}
The function $\Psi (\cdot , \xi_{ 1 } )$ after the first iteration is displayed in Figure \ref{figrefine}, where we used the same starting design as in
 \cite{BraessDette2013}.
 The support  points  of  $ \ \xi_{ 1 }$ are shown as circles and we can see that function $\Psi(x, \xi_{1})$ attains one and the same value, which is represented with dotted line, for all support points.
We finally note that the algorithm proposed in  \cite{BraessDette2013} needs $9$ iterations to find a design with the same efficiency.

\begin{figure}[h]
\centering
  \includegraphics[width=45mm]{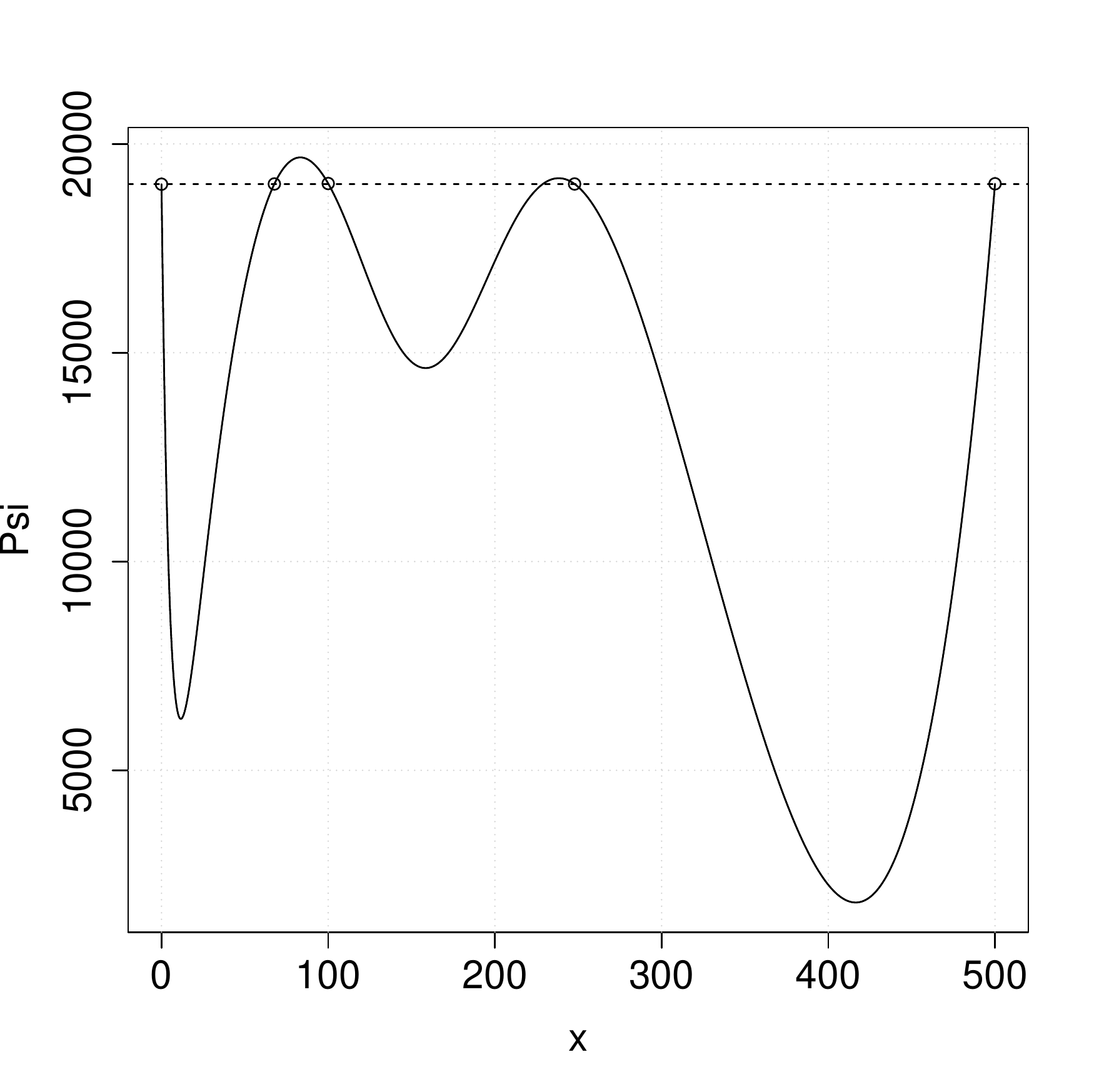}
\caption{\it \label{figrefine}  The function  $\Psi (\cdot ,  \xi_{1} )$ after the first iteration of  Algorithm   \ref{algorithm:new} }
\end{figure}

We now investigate Bayesian $T$-optimal discriminating designs for a similar situation.  For the sake of a transparent
representation we  only specify a prior distribution of the four-dimensional parameter $\overline{\theta}_4$
for the calculation of the discriminating design, while  $\overline{\theta}_2$ and $\overline{\theta}_3$  are considered as fixed.
In order to obtain a design which is robust with respect to model misspecification
we chose a prior discrete prior with $81$ points in $\R^4$. More precisely, the support points  of the prior distribution are given
by the points
\begin{align} \label {prior1}
\big\{\mu_{e_1,e_2,e_3,e_4}~|~e_1,e_2,e_3,e_4 \in \{-1,0,1\}\big\},
\end{align}
 where
 \begin{align*}
  \mu_{e_1,e_2,e_3,e_4} & = (\mu_1 + e_{1} \sigma, \mu_2 + e_{2} \sigma,
 \mu_3 + e_{3} \sigma, \mu_4 + e_{4} \sigma),\\
 \mu &=(\mu_1,\mu_2,\mu_3,\mu_4)=  (49.62, 290.51, 150, 45.51),
 \end{align*}
 and different values for $\sigma^2$ are considered.
  The weights at the corresponding points are proportional (normalized such that their sum is $1$) to
\begin{align} \label {prior2}
 \frac{1}{(2\pi\sigma ^2)^2} \exp \Big(\frac{ || \mu_{e_1,e_2,e_3,e_4} - \mu ||_2^2}{2\sigma^2} \Big)~,~~e_1,e_2,e_3,e_4 \in \{-1,0,1\},
\end{align}
  where $||\cdot ||_2$ denotes the Euclidean norm.   The resulting Bayesian optimality criterion \eqref{2.6}
  consist of $246$ model comparisons.  In this case the method of \cite{BraessDette2013} fails to find the Bayesian $T$-optimal discriminating design. Bayesian $T$-optimal discriminating designs have been calculated by the new
Algorithm   \ref{algorithm:new} for various values of $\sigma^2$ and the results are shown in Table \ref{tab2}. A typical determination of the
  optimal design takes between $0.09$ seconds  (in the case  $\sigma^2=0$) and  $7.8$ seconds  (in the case  $\sigma^2=37^2$)  CPU time on a standard PC.
  The algorithm using the procedure described in Section \ref{sec32}
in Step 2 requires  between
$0.75$  seconds  (in the case  $\sigma^2=0$)   and  $37.1$  seconds  (in the case  $\sigma^2=37^2$)    CPU  time.
  For small values the Bayesian $T$-optimal discriminating designs are  supported at $4$ points including the boundary
  of the design space. The smaller (larger) interior support point is increasing (decreasing)
  if $\sigma^2$ is increasing. For larger values of $\sigma^2$ even the number of support points of the optimal
  design increases. For example, if $\sigma^2=35^2$ or $37^2$ the Bayesian $T$-optimal discriminating design
  has $5$ or $6$ points (including the boundary points of the design space). These observations are in line with the
  theoretical finding of  \cite{detbra2007}  who showed that the number of support points of Bayesian $D$-optimal designs can become
  arbitrarily large with an increasing variability in the prior distribution. The corresponding functions from the equivalence Theorem \ref{thm1}
  are shown in Figure \ref{fig2}.

\begin{table}[h]
\begin{center}
\begin{tabular}{||c||c||c||c||}
\hline
\hline
$\sigma^2$ & optimal design & $\sigma^2$ & optimal design\\
\hline
\hline
0 & $\begin{matrix}
0.000 & 78.783 & 241.036 & 500.0 \\
0.255 & 0.213 & 0.357 & 0.175
\end{matrix}$ & $33^2$ & $\begin{matrix}
0.000 & 92.692 & 222.735 & 500.0 \\
0.260 & 0.240 & 0.344 & 0.156
\end{matrix}$\\
\hline
$20^2$ & $\begin{matrix}
0.000 & 84.467 & 234.134 & 500.0 \\
0.257 &  0.225 &  0.351 & 0.167
\end{matrix}$ & $35^2$ & $\begin{matrix}
0.000 & 91.743 & 129.322 & 221.118 & 500.0 \\
0.260 & 0.214 & 0.036 & 0.336 & 0.154
\end{matrix}$\\
\hline
$30^2$ & $\begin{matrix}
0.000 & 91.029 & 225.713 & 500.0 \\
0.259 & 0.237 & 0.345 & 0.159
\end{matrix}$ & $37^2$ & $\begin{matrix}
0.000 & 89.881 & 129.590 & 170.306 & 220.191 & 500.0 \\
0.260 & 0.170 & 0.091 & 0.019 & 0.310 & 0.150
\end{matrix}$\\
\hline
\hline
\end{tabular}
\caption{\it Bayesian $T$-optimal discriminating designs for the models in \eqref{example2}.  The weights
in the criterion \eqref{2.5}  are given by $p_{i,j}=1/6$; $1 \leq i < j \leq 4$ and the support  and  masses of the
prior distribution are
defined by \eqref{prior1} and \eqref{prior2}, respectively.   \label{tab2}}
\end{center}
\end{table}

\begin{figure}[h]
\centering
  {\includegraphics[width=45mm]{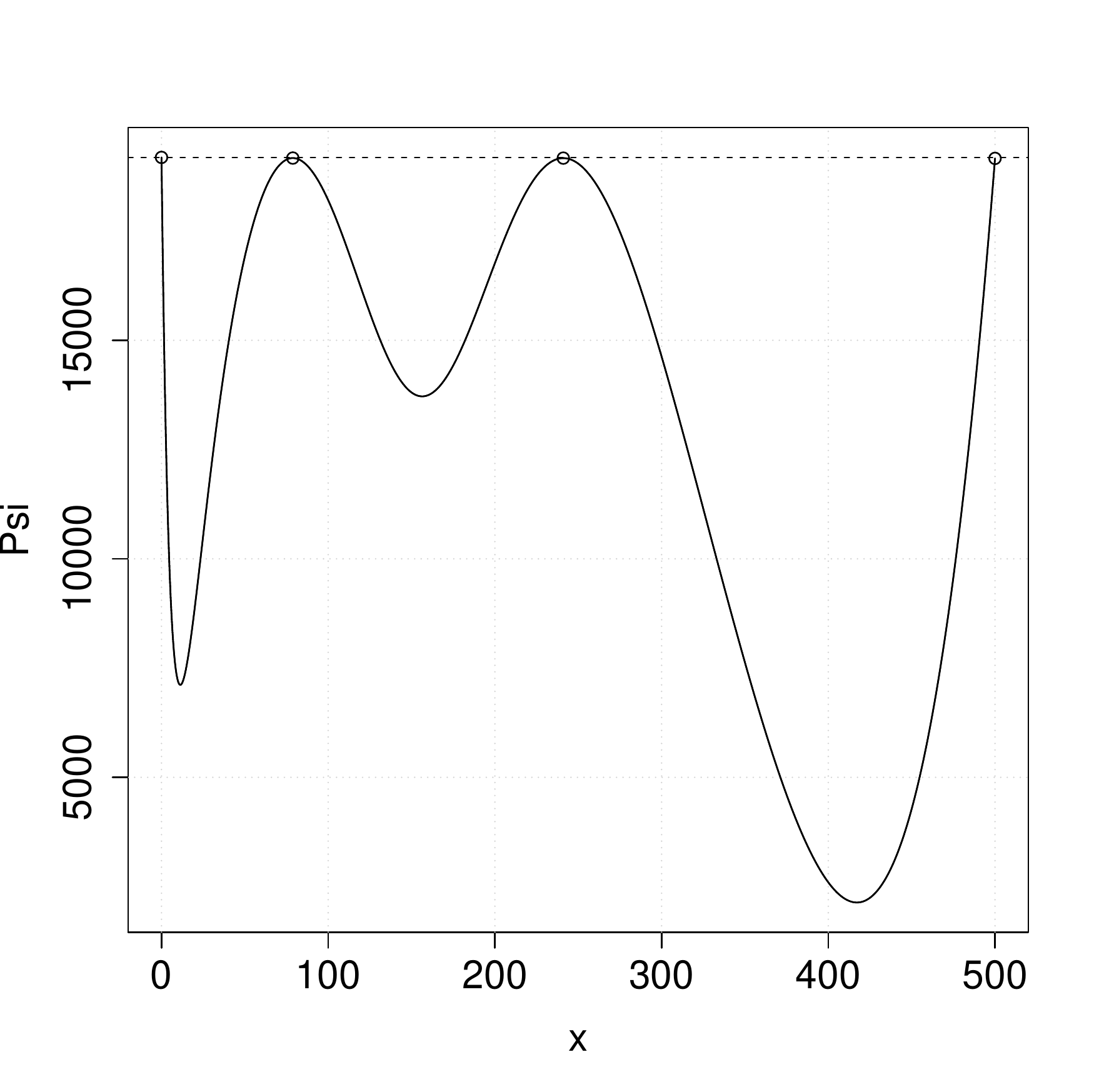}}
  {\includegraphics[width=45mm]{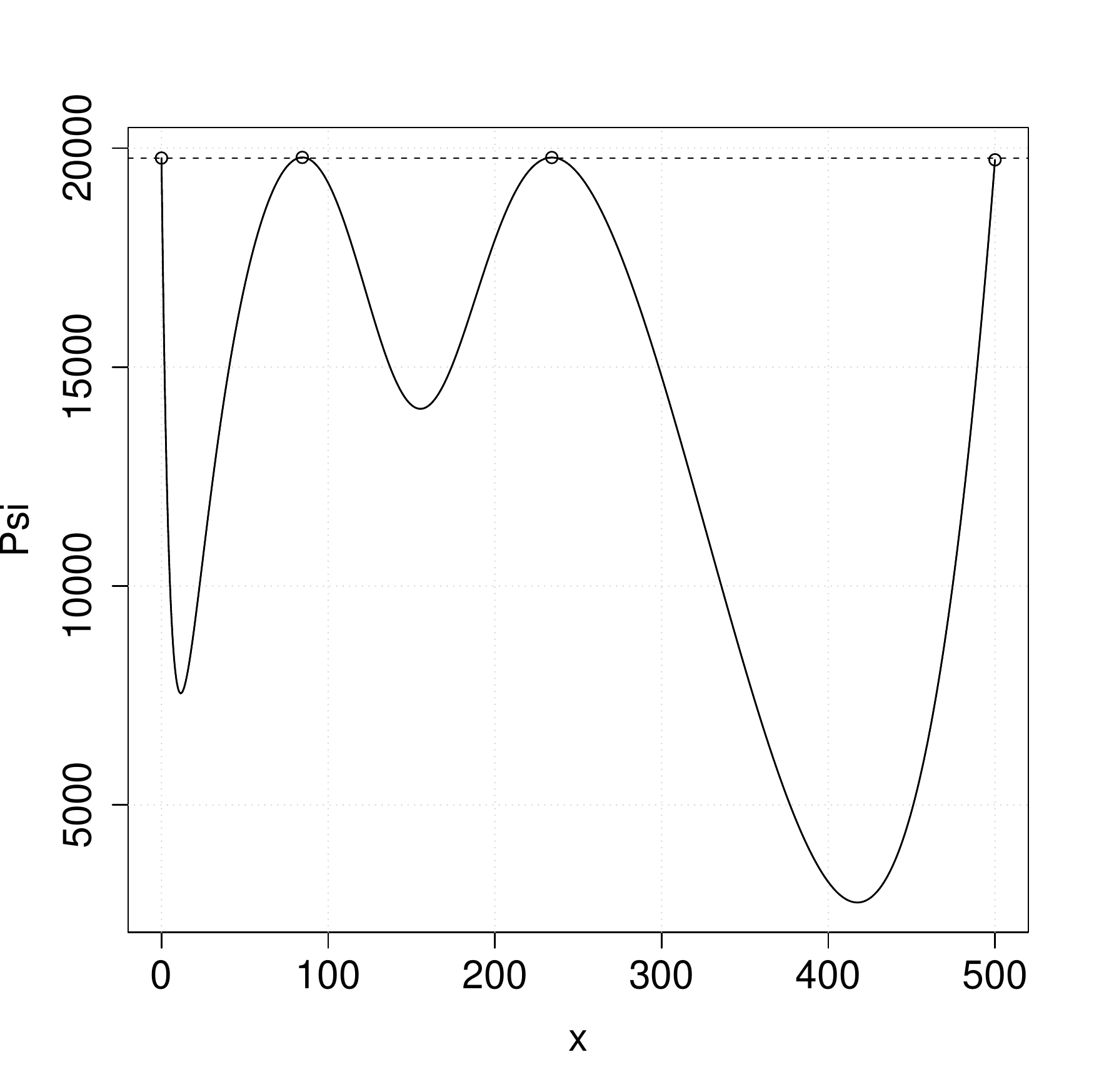}}
  {\includegraphics[width=45mm]{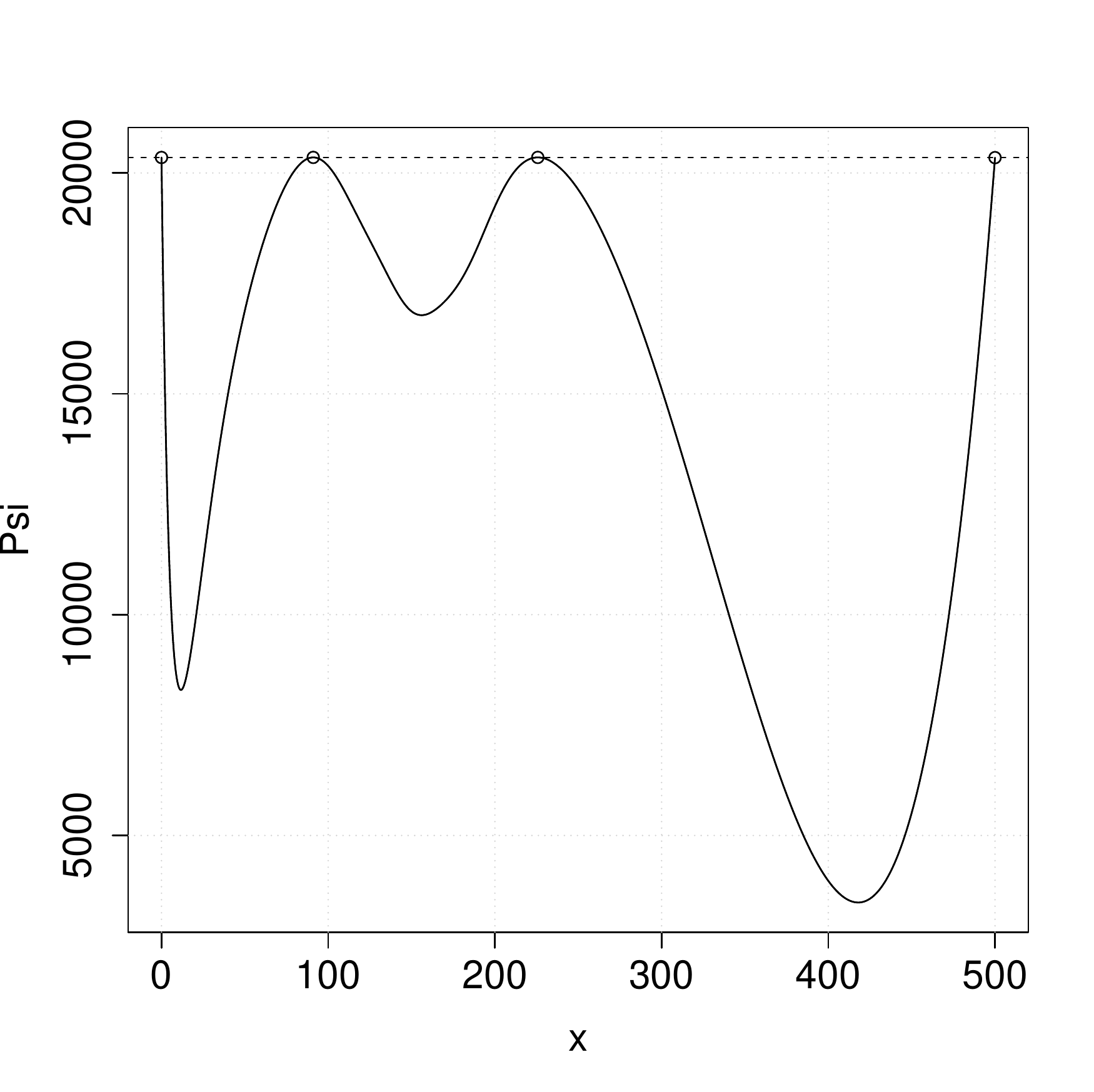}} \\
$\sigma^2 = 0$~~~~~~~~~~~~~~~~~~~~~~~~$\sigma^2 = 20$ ~~~~~~~~~~~~~~~~~~~ $\sigma^2 = 30$  \\
  {\includegraphics[width=45mm]{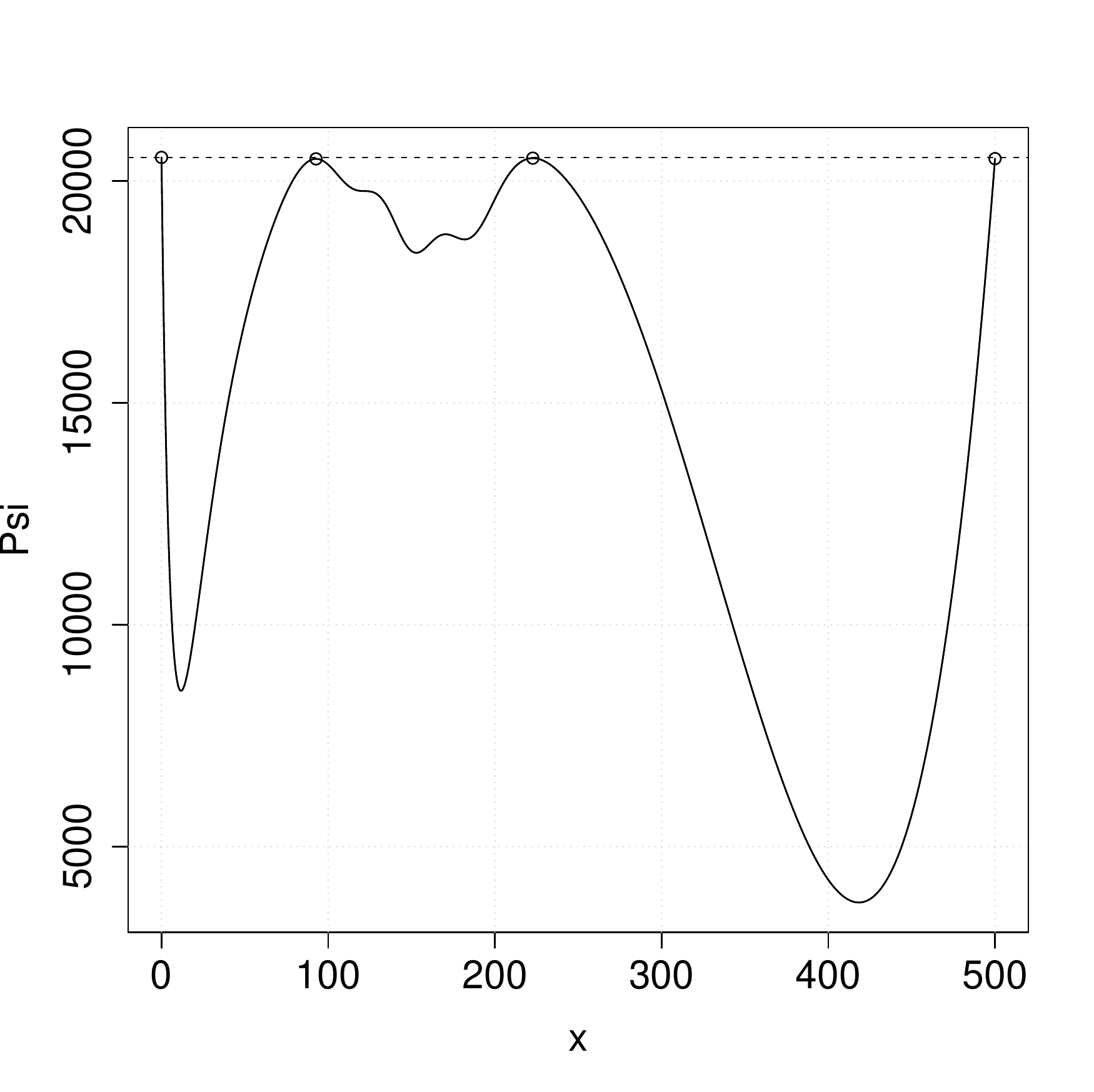}}
  {\includegraphics[width=45mm]{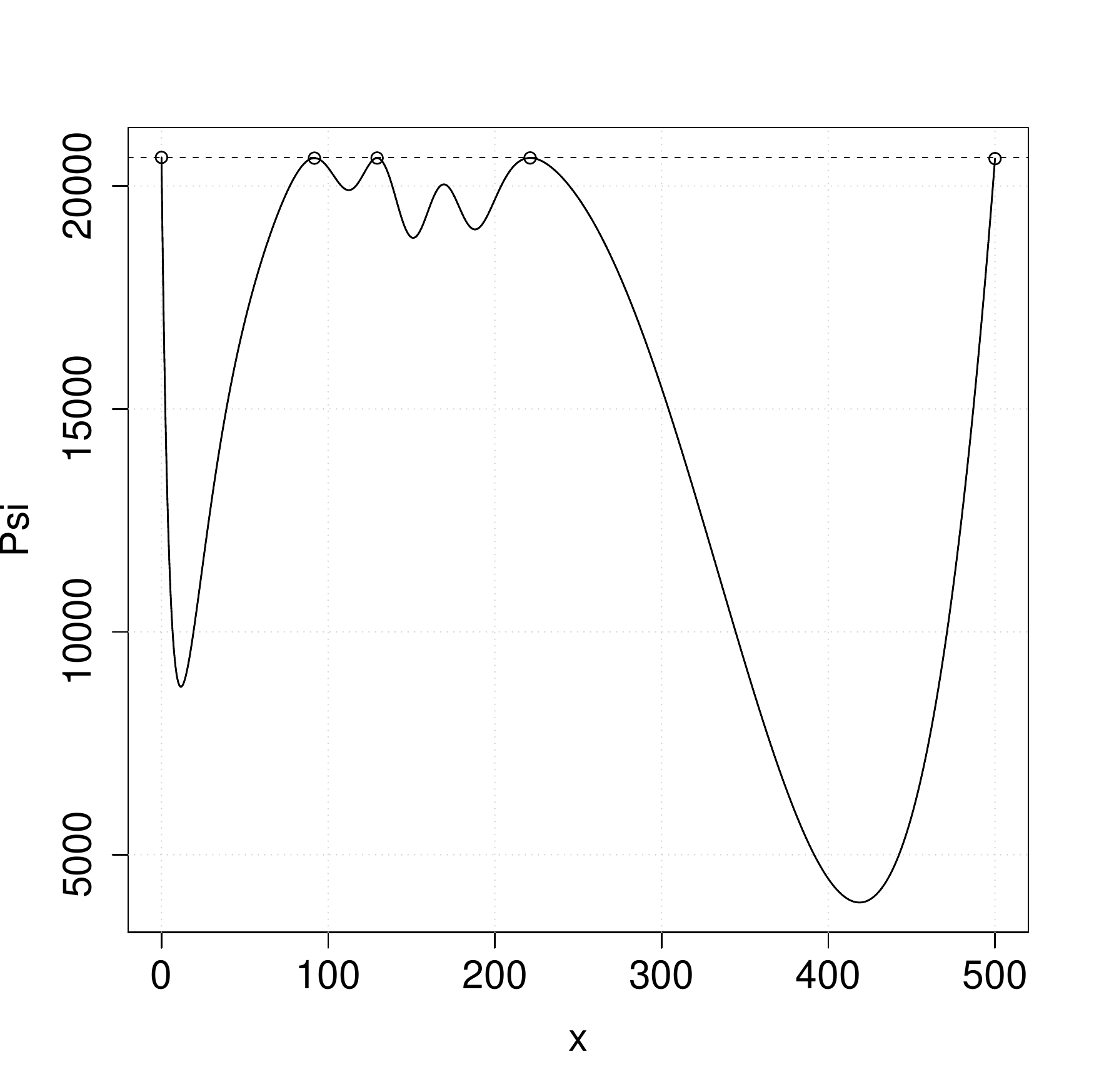}}
 {\includegraphics[width=45mm]{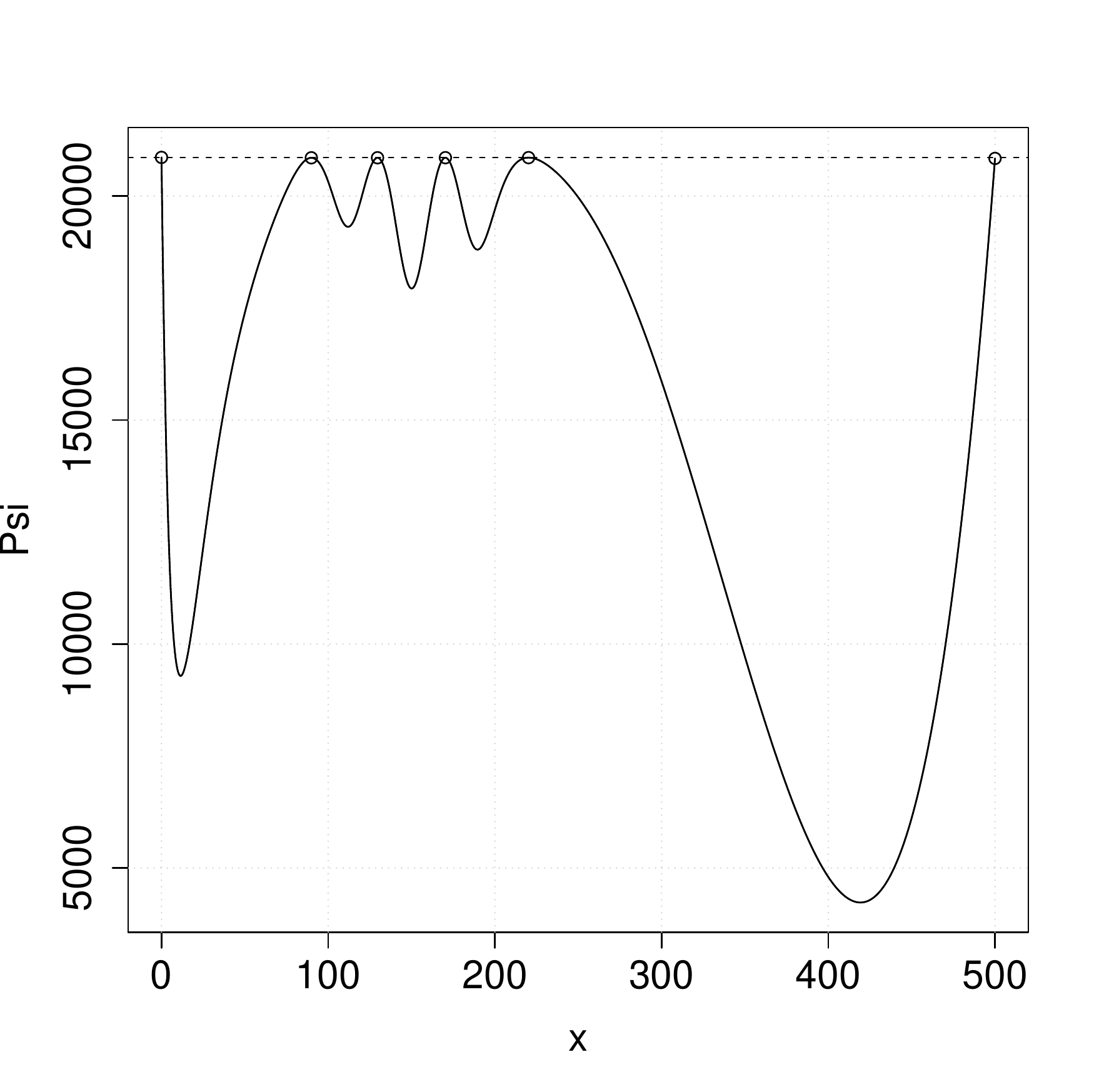}}
$\sigma^2 = 33$~~~~~~~~~~~~~~~~~~~~~~~~$\sigma^2 = 35$ ~~~~~~~~~~~~~~~~~~~ $\sigma^2 =37$ \\
\caption{\it \label{fig2}  The function on the left-hand side of inequality
\eqref{equiv} in the equivalence Theorem \ref{thm1} for the numerically calculated Bayesian $T$-optimal discriminating designs.
The  competing regression models are given in \eqref{example2}.}
\end{figure}

\bigskip

{\bf Acknowledgements.} Parts of this work were done during a visit of the
second author at the Department of Mathematics, Ruhr-Universit\"at
Bochum, Germany.
The authors would like to thank
M. Stein who typed this manuscript with considerable technical
expertise.
The work of H. Dette and V. Melas was supported by the Deutsche
Forschungsgemeinschaft (SFB 823: Statistik nichtlinearer dynamischer Prozesse, Teilprojekt C2).
The research of H. Dette reported in this publication was also partially supported by the National Institute of
General Medical Sciences of the National Institutes of Health under Award Number R01GM107639.
The content is solely the responsibility of the authors and does not necessarily
 represent the official views of the National
Institutes of Health.  V. Melas was also partially supported by Russian Foundation of Basic Research (Project 12.01.00747a).

\section{Proofs} \label{sec6}

\subsection{An auxiliary result} \label{sec6a}

\begin{lemma} \label{lemma1}
Let $\varphi(v,y)$ be a  twice continuously differentiable function of two  variables $v\in \mathcal{V} \subset \R^k$  and $y \in \mathcal{Y},  $ where   $ \mathcal{Y}   $ is a compact set.
Denote by ${ \mathcal{Y}^*  }$ the set of all points where the minimum
$\min_{y \in  \mathcal{Y}  } \varphi(v,y)$ is attained and let  $q \in \R^k$ be an arbitrary direction.
Then
\begin{align}
\frac{\partial \min_{y \in { \mathcal{Y}^*  }} \varphi(v,y)}{\partial q}= \min_{y \in{ \mathcal{Y} ^* }}\frac{\partial\varphi(v,y)}{\partial q}.
\label{formula_diff}
\end{align}
\end{lemma}

\begin{proof} See \cite{Pshenichny1971}, p. 75.
\end{proof}

\subsection{Proofs} \label{sec6b}

{\bf Proof of Theorem \ref{thm2.2}.} Assume without loss of generality that $p_{i,j} >0$ for all $i,j=1,\dots,\nu$. Let $\xi^* $ denote any locally $T$-optimal discriminating design and
let $\theta=(\theta_{i,j})_{i,j=1,\dots,\nu}$ denote  the vector consisting of all  $\theta_{i,j} \in \Theta_{i,j}(\xi^*) $.
We introduce the function
\begin{align} \label{varphi}
\varphi(x, \theta) & = \sum_{i,j = 1}^{\nu} p_{i,j} \big[ \eta_i(x,\overline{\theta}_{i}) - \eta_j(x,{\theta}_{i,j}) \big]^2
\end{align}
and   consider the product measure
\begin{align} \label{prodmeas}
\mu(d\theta)= \prod_{i,j=1,\ldots,\nu} \mu_{i,j}(d\theta_{i,j}),
\end{align}
 where $\mu_{ij}$ are measures on the sets  ${\Theta}_{i,j}^*(\xi^*)$ defined by (\ref{thetamin}).
Similarly, we define  $\mu^* (d\theta)= \prod_{i,j=1,\ldots,\nu} \mu_{i,j}^* (d\theta_{i,j})$ as the product measure of the measures $\mu^*_{i,j}$
in Theorem \ref{thm1}. From this result we have
\begin{eqnarray*}
T_{\mathrm{P}}(\xi^*) &\geq& \sup_\zeta \int_\mathcal{X}  \int_{\Theta^*(\xi^*)}  \varphi(x, \theta) \mu^* (d\theta) \zeta(dx) \\
& \geq &\inf_\mu \sup_\zeta \int_\mathcal{X} \int_{\Theta^*(\xi^*)} \varphi(x, \theta) \mu (d\theta) \zeta(dx) = \sup_\zeta \inf_\mu \int_\mathcal{X} \int_{\Theta^*(\xi^*)}  \varphi(x, \theta) \mu (d\theta) \zeta(dx),
\end{eqnarray*}
where the  $\sup$  and $\inf$  are calculated in the class of  designs $\zeta$ on $\mathcal{X}$ and
product measures  $\mu$   on  $\Theta^*(\xi^*)=\otimes_{i,j=1}^{\nu} \Theta_{i,j(\xi^*)}^* $, respectively.
It now follows   that the characterizing inequality (\ref{equiv}) in Theorem \ref{thm1}  is equivalent to the inequality
$$
\sup_{\zeta}  Q(\zeta,\xi^*) \leq T_{\mathrm{P}}(\xi^*).
$$
Consequently, any non-optimal design must satisfy the opposite inequality. \hfill $\Box$

\bigskip

\textbf{Proof of Corollary \ref{thm1a}:}
Let $\xi$ denote a design such that    $T_\mathrm{P}(\xi)>0$ and  recall the definition of  the set $\Theta_{ij}^*(\xi)$
in \eqref{thetamin}.  We consider for a vector
$\theta= (\theta_{i,j})_{ i,j=1,\ldots,\nu} \in \Theta^*(\xi)= \otimes_{ i,j=1,\ldots,\nu}  \Theta_{i,j}^*(\xi)$, the function $\varphi$ is defined in \eqref{varphi}
and product measures $\mu(d\theta)$ of the form \eqref{prodmeas}
on $\Theta^*(\xi)$.
Now the  well  known  minimax theorem and the definition of the function $Q$ in \eqref{qfct}  yields
\begin{eqnarray*}
\max_{x \in \mathcal{X}}\Psi(x,\xi)&=&\inf_\mu \max_{x \in \mathcal{X}} \int_{\Theta^*(\xi)}  \varphi(x, \theta) \mu (d\theta) =
\inf_\mu \sup_\zeta \int_{\mathcal X}  \int_{\Theta^*(\xi)}  \varphi(x, \theta) \mu (d\theta) \zeta (d x)   \\
&=&\sup_\zeta \inf_\mu
 \int_{\mathcal X}  \int_{\Theta^*(\xi)}
\varphi(x, \theta) \mu (d\theta) \zeta (dx) =
\sup_\zeta \inf_{\theta \in \Theta^*(\xi) }\int \varphi(x, \theta) \zeta (dx) = \sup_\zeta Q(\zeta ,\xi),
\end{eqnarray*}
where the infimum is calculated with respect to all measures $\mu$ of the form \eqref{prodmeas} and the
supremum is calculated with respect to all experimental designs $\zeta$ on  $\mathcal{X}$.
Note that $\mathcal X$ is compact by assumption and it can be checked that the set $\Theta^*(\xi)$ is
also compact as a closed subset of a compact set. Consequently all
  suprema and infima are achieved and   there exists a
design $\zeta^*$ supported at the set of local maxima of the function $\Psi(x,\xi)$,
 such that
$$
Q(\zeta^*,\xi)=\sup_\zeta Q(\zeta,\xi) = \max_{x \in \mathcal{X}} \Psi(x,\xi).
$$
The assertion of Corollary  \ref{thm1a} now  follows from Theorem \ref{thm2.2}. \hfill $\Box$

\medskip

\textbf{Proof of Theorem \ref{thm2}:}
Obviously, the inequality
\begin{align*}
T_{\mathrm{P}}(\{\mathcal{S}_{[s]},\omega_{[s]}\}) \leq T_{\mathrm{P}}(\{\mathcal{S}_{[s+1]},\omega_{[s+1]}\}) 
\end{align*}
holds  for all $s$
as optimization with respect to   $\omega$ occurs on a larger set. Moreover, the sequence  $T_{\mathrm{P}}(\xi_s)$ is bounded from above by
$T_{\mathrm{P}}(\xi^*)$ and  has a limit, which is denoted by   $T^{**}_{\mathrm{P}}$.
Consequently, there exists a subsequence of designs, say $ \xi_{s_j},j=1,2, \ldots$ converging to a design, say $\xi^{**}$.
Note that $T_P$ is upper semi-continuous as the infimum of continuous functions, which implies
 $T_{\mathrm{P}}(\xi^{**})= T^{**}_{\mathrm{P}}$. Now, assume that  $T_{\mathrm{P}}(\xi^{**}) < T_{\mathrm{P}}(\xi^{*})$,  then $\xi^{**}$ is not locally $T$-optimal and by
 Theorem \ref{thm2.2} there exists a constant $\delta > 0$ such that
\begin{align*}
{\sup_{\zeta} Q(\zeta,\xi^{**}) - T_{\mathrm{P}}(\xi^{**}) = 2 \delta},
\end{align*}
where the function $Q$ is defined in \eqref{qfct}.
Therefore for sufficiently large $j$, say,  $j \ge  N$ we obtain (using again the lower semi-continuity of
$ \sup_{\zeta} Q(\zeta ,\xi)$) that
\begin{align*}
\sup_{\zeta} Q(\zeta,{\xi}_{s_j}) - T_{\mathrm{P}}(\xi_{s_j})> \delta,
\end{align*}
whenever $j\ge N$. {Note that by construction the sequence $(T_{\mathrm{P}}(\xi_s))_{s \in \N}$ is increasing and therefore
\begin{equation}\label{ungl}
T_{\mathrm{P}}({\xi}_{s_{j+1}}) - T_{\mathrm{P}}(\xi_{s_{j}})  \geq T_{\mathrm{P}}(\xi_{s_{j} +1})- T_{\mathrm{P}}(\xi_{s_{j}}).
\end{equation}
In order to estimate the right hand side
 we  consider  for $j \ge N$ and $\alpha \in [0,1]$ the design
\begin{align*}
\tilde{\xi}_{s_{j+1}}(\alpha) = (  1 - \alpha ) {\xi}_{s_j} + \alpha \zeta_{j},
\end{align*}
where
   $\zeta_{j}$ is the measure for  which the function
$
Q( \zeta,{\xi}_{s_j} )
$
attains its maximal value in the class of all experimental designs supported at the local maxima of the function $\Psi(x,{\xi}_{s_j})$, and
define 
\begin{align*}
\alpha_{s_{j+1}} = \argmax_{0 \leq \alpha \leq 1} T_{\mathrm{P}}(\tilde{\xi}_{s_{j+1}}(\alpha)).
\end{align*}
By  construction of ${\xi}_{s_{j}+1}$ is the best design supported  at $\mbox{supp}({\xi}_{s_j}) \cup \mbox{supp}(\zeta_j)$,    and \eqref{ungl} yields
\begin{equation} \label{ungl1}
T_{\mathrm{P}}({\xi}_{s_{j+1}}) \geq T_{\mathrm{P}}({\xi}_{s_{j}+1}) \geq  T_{\mathrm{P}}(\tilde {\xi}_{s_{j+1}}(\alpha_{s_{j+1}})).
\end{equation}
We introduce the notations  $
h(j,\alpha)= T_{\mathrm{P}}(\tilde {\xi}_{s_{j}}(\alpha))
$,
and note  that
\begin{align*}
\frac{\partial T_{\mathrm{P}}(\tilde {\xi}_{s_{j+1}}(\alpha))}{\partial\alpha}\Big|_{\alpha=0}=
Q(\zeta_j, {\xi}_{s_j})
-T_{\mathrm{P}}
({\xi}_{s_j}) =
\sup_\zeta  Q (\zeta, {\xi}_{s_j})
-T_{\mathrm{P}}({\xi}_{s_j})>\delta.
\end{align*}
A  Taylor expansion gives 
\begin{align*}
h(j+1,\alpha_{s_{j+1}})-h(j+1,0)  = & \max_{\alpha \in [0,1] } \big [  T_{\mathrm{P}}( \tilde{\xi}_{s_{j+1}}(\alpha)) - T_{\mathrm{P}}( \tilde{\xi}_{s_{j+1}}(0 ))\big]  \\
 \geq &  \max_{\alpha \in [0,1] }\Big[\alpha \frac{\partial T_{\mathrm{P}}( \tilde{\xi}_{s_{j+1}}(\alpha))}{\partial\alpha}\Big|_{\alpha=0}
 -  \frac{1}{2}\alpha^2K \Big]
 \notag > \max_{\alpha  \in [0,1]  }\Big[\alpha \delta - \frac{1}{2}\alpha^2K\Big] = \frac { \delta^2}{2K},
\end{align*}
where   $K$ is an absolute upper bound of the second derivative. Therefore it  follows  from  \eqref{ungl1}  that 
\begin{align*}
T_{\mathrm{P}}(\xi_{s_{j+1}}) - T_{\mathrm{P}}(\xi_{s_j}) &\geq T_{\mathrm{P}}(\xi_{s_j+1}) - T_{\mathrm{P}}(\xi_{s_j}) \\
&\geq T_{\mathrm{P}}(\tilde\xi_{s_{j+1}}(\alpha_{s_{j+1}})) - T_{\mathrm{P}}(\xi_{s_j}) = h(j+1,\alpha_{s_{j+1}}) - h(j+1,0) \geq \frac{\delta^2}{2K}.
\end{align*}
which gives for $L > N+1$
\begin{align*}
T_{\mathrm{P}}(\xi_{s_{L}}) - T_{\mathrm{P}}(\xi_{s_N})  = \sum_{j = N}^{L - 1} \big [ T_{\mathrm{P}}(\xi_{s_{j+1}}) - T_{\mathrm{P}}(\xi_{s_j}) \big ]  \geq \left[ L - N \right] \frac{\delta^2}{2K}.
\end{align*}
The left hand side of this  inequality converges to the finite value $T(\xi^{**}) - T(\xi_{s_N})$ as   $L \to \infty$, while   the right hand side converges
to  infinity. Therefore we obtain a contradiction to our
assumption  $T_{\mathrm{P}}(\xi^{**})  < T_{\mathrm{P}}(\xi^{*})$, which proves the assertion of Theorem \ref{thm2}.

\bigskip

\textbf{Proof of Lemma \ref{lem2}:}
Fix $t \in \{ 1,\dots,n \}$ and note that $w_t=1 - \sum^n_{\ell =1, \ell \neq t} w_\ell$.
 Under Assumptions \ref{assum1} and \ref{assum2}  we obtain by formula (\ref{formula_diff})
\begin{align*}
\frac{\partial g(\omega)}{\partial \omega_k} &= \sum_{i,j = 1}^{\nu} p_{i,j} \big[ \eta_i(x_k,\overline{\theta}_{i}) - \eta_j(x_k,\widehat{\theta}_{i,j}(\omega)) \big]^2 - \sum_{i,j = 1}^{\nu} p_{i,j}
\big[ \eta_i(x_t,\overline{\theta}_{i}) - \eta_j(x_t,\widehat{\theta}_{i,j}(\omega)) \big]^2
\end{align*}
The condition $
\frac{\partial g(\omega)}{\partial \omega_k} = 0, \; k = 1,\dots,n, \; k \neq t
$
is the necessary condition for weight optimality and consequently it follows from the definition of the function
$
\Psi(x,\overline{\xi}_{s+1})
$
that this function attains one and the same value for all support points of the design $\overline\xi_{s+1}$.

\bigskip

\textbf{Proof of Theorem \ref{conv_theorem}:}
The proof is similar to the proof of Theorem \ref {thm2}. Denote
\begin{align*}
h(\gamma,\alpha)= g(\overline{\omega}_{(\gamma)}(\alpha)),
\end{align*}
where the vector $\overline{\omega}_{(\gamma)} (\alpha^*)$ is calculated at the $\gamma$th iteration.
 Since the sequence $ g(\omega_{(\gamma)})$ is bounded and increasing (by  construction)  it converges  to some limit, say  $ g^{**}$.
Consequently  there exists a subsequence of vector of weights, say $\overline \omega_{(\gamma_j)},j=1,2, \ldots$ converging to a vector, say $\overline \omega^{**}$.
 Note that $ g$ is upper semi-continuous as the infimum of continuous functions, which implies
 $ g(\overline{\omega}^{**})= g^{**}$. Now, assume that  $ g(\overline{\omega}^{**})  <   g(\omega^{*})$,  then it follows
by an application of Theorem  \ref{thm1}   with $\mathcal{X}=\{x_1,\ldots , x_n\}$
 that there exists  a constant $\delta > 0$ such that
\begin{align*}
\frac{\partial  g(\overline{\omega}(\alpha))}{\partial \alpha} \Bigl|_{\alpha=0} = 2\delta > 0 .
\end{align*}
Here  the vector  $\overline{\omega}(\alpha)$ is defined in the same way as $\overline{\omega}_{(\gamma )} (\alpha)$, where
 $\omega_{(\gamma)}$ is replaced by $\omega=\omega^{**}$.
Therefore for sufficiently large $j$, say,  $j \ge  N$ we obtain (using the lower semi-continuity of
${g}$) that
$
h(\gamma_j,0)> \delta ,
$
and  a Taylor expansion yields
\begin{align*}
h(\gamma_{j+1},\alpha^*_{(\gamma_{j+1})})-h(s_j,\alpha^*_{(\gamma_j)}))\geq
\max_\alpha \Big( \alpha \frac{\partial  g(\overline\omega(\alpha))}{\partial \alpha}
-  \frac{1}{2}\alpha^2K  \Big) =\frac{\delta^2}{2K},
\end{align*}
where  $\alpha^*_{(\gamma_j)}$ is the value $\alpha^*$ from the $\gamma_j$th iteration and
 $K$ is an absolute upper bound of the second derivative.
Using the  same arguments as in the proof of Theorem \ref {thm2} we obtain a  contradiction, which proves the assertion of the theorem.

\bigskip

\setstretch{1.25}
\setlength{\bibsep}{1pt}
\begin{small} \footnotesize
 \bibliographystyle{apalike}
\itemsep=0.5pt
\bibliography{model}

\begin{thebibliography}{}

\bibitem[Atkinson et~al., 2007]{atkinson2007}
Atkinson, A., Donev, A., and Tobias, R. (2007).
\newblock {\em Optimum Experimental Designs, with {SAS} (Oxford Statistical
  Science Series)}.
\newblock Oxford University Press, USA, 2nd edition.

\bibitem[Atkinson, 2008]{atkinson2008}
Atkinson, A.~C. (2008).
\newblock Examples of the use of an equivalence theorem in constructing optimum
  experimental designs for random-effects nonlinear regression models.
\newblock {\em Journal of Statistical Planning and Inference},
  138(9):2595--2606.

\bibitem[Atkinson et~al., 1998]{atkbogbog1998}
Atkinson, A.~C., Bogacka, B., and Bogacki, M.~B. (1998).
\newblock ${D}$- and ${T}$-optimum designs for the kinetics of a reversible
  chemical reaction.
\newblock {\em Chemometrics and Intelligent Laboratory Systems}, 43:185--198.

\bibitem[Atkinson and Fedorov, 1975a]{atkfed1975a}
Atkinson, A.~C. and Fedorov, V.~V. (1975a).
\newblock The designs of experiments for discriminating between two rival
  models.
\newblock {\em Biometrika}, 62:57--70.

\bibitem[Atkinson and Fedorov, 1975b]{atkfed1975b}
Atkinson, A.~C. and Fedorov, V.~V. (1975b).
\newblock Optimal design: {E}xperiments for discriminating between several
  models.
\newblock {\em Biometrika}, 62:289--303.

\bibitem[Braess and Dette, 2007]{detbra2007}
Braess, D. and Dette, H. (2007).
\newblock On the number of support points of maximin and {B}ayesian
  {$D$}-optimal designs in nonlinear regression models.
\newblock {\em Annals of Statistics}, 35:772--792.

\bibitem[Braess and Dette, 2013]{BraessDette2013}
Braess, D. and Dette, H. (2013).
\newblock Optimal discriminating designs for several competing regression
  models.
\newblock {\em Annals of Statistics}, 41(2):897--922.

\bibitem[Chaloner and Verdinelli, 1995]{chaver1995}
Chaloner, K. and Verdinelli, I. (1995).
\newblock Bayesian experimental design: {A} review.
\newblock {\em Statistical Science}, 10(3):273--304.

\bibitem[Chernoff, 1953]{chernoff1953}
Chernoff, H. (1953).
\newblock Locally optimal designs for estimating parameters.
\newblock {\em Annals of Mathematical Statistics}, 24:586--602.

\bibitem[Dette, 1997]{dette1997}
Dette, H. (1997).
\newblock Designing experiments with respect to ``standardized'' optimality
  criteria.
\newblock {\em Journal of the Royal Statistical Society, Ser. B}, 59:97--110.

\bibitem[Dette and Haller, 1998]{dethal1998}
Dette, H. and Haller, G. (1998).
\newblock Optimal designs for the identification of the order of a {F}ourier
  regression.
\newblock {\em Annals of Statistics}, 26:1496--1521.

\bibitem[Dette et~al., 2012]{detmelshp2012}
Dette, H., Melas, V.~B., and Shpilev, P. (2012).
\newblock T-optimal designs for discrimination between two polynomial models.
\newblock {\em Annals of Statistics}, 40(1):188--205.

\bibitem[Dette et~al., 2013]{DetteMelasShpilev2013}
Dette, H., Melas, V.~B., and Shpilev, P. (2013).
\newblock Robust {$T$}-optimal discriminating designs.
\newblock {\em Annals of Statistics}, 41(4):1693--1715.

\bibitem[Foo and Duffull, 2011]{fooduf2011}
Foo, L.~K. and Duffull, S. (2011).
\newblock Optimal design of pharmacokinetic-pharmacodynamic studies.
\newblock In Bonate, P.~L. and Howard, D.~R., editors, {\em Pharmacokinetics in
  Drug Development, Advances and Applications}. Springer.

\bibitem[Han and Chaloner, 2003]{hancha2003}
Han, C. and Chaloner, K. (2003).
\newblock {$D$}- and {$c$}-optimal designs for exponential regression models
  used in pharmacokinetics and viral dynamics.
\newblock {\em Journal of Statistical Planning and Inference}, 115:585--601.

\bibitem[Kiefer, 1974]{kiefer1974}
Kiefer, J. (1974).
\newblock General equivalence theory for optimum designs (approximate theory).
\newblock {\em Annals of Statistics}, 2(5):849--879.

\bibitem[L\'{o}pez-Fidalgo et~al., 2007]{loptomtra2007}
L\'{o}pez-Fidalgo, J., Tommasi, C., and Trandafir, P.~C. (2007).
\newblock An optimal experimental design criterion for discriminating between
  non-normal models.
\newblock {\em Journal of the Royal Statistical Society, Series B},
  69:231--242.

\bibitem[Pinheiro et~al., 2006]{pinbrebra2006}
Pinheiro, J., Bretz, F., and Branson, M. (2006).
\newblock Analysis of dose-response studies: {M}odeling approaches.
\newblock In Ting, N., editor, {\em Dose Finding in Drug Development}, pages
  146--171. Springer-Verlag, New York.

\bibitem[Pronzato and Walter, 1985]{pronwalt1985}
Pronzato, L. and Walter, E. (1985).
\newblock Robust experimental design via stochastic approximation.
\newblock {\em Mathematical Biosciences}, 75:103--120.

\bibitem[Pshenichny, 1971]{Pshenichny1971}
Pshenichny, B.~N. (1971).
\newblock {\em Necessary Conditions of an Extremum}.
\newblock Marcel Dekker, New York.

\bibitem[Pukelsheim, 2006]{pukelsheim2006}
Pukelsheim, F. (2006).
\newblock {\em Optimal Design of Experiments}.
\newblock SIAM, Philadelphia.

\bibitem[Ratkowsky, 1990]{ratkowsky1990}
Ratkowsky, D. (1990).
\newblock {\em Handbook of Nonlinear Regression Models}.
\newblock Dekker, New York.

\bibitem[Song and Wong, 1999]{songwong1999}
Song, D. and Wong, W.~K. (1999).
\newblock On the construction of $g_{rm}$-optimal designs.
\newblock {\em Statistica Sinica}, 9:263--272.

\bibitem[Stigler, 1971]{stigler1971}
Stigler, S. (1971).
\newblock Optimal experimental design for polynomial regression.
\newblock {\em Journal of the American Statistical Association}, 66:311--318.

\bibitem[Tommasi, 2009]{Tommasi09}
Tommasi, C. (2009).
\newblock Optimal designs for both model discrimination and parameter
  estimation.
\newblock {\em Journal of Statistical Planning and Inference}, 139:4123--4132.

\bibitem[Tommasi and L\'{o}pez-Fidalgo, 2010]{tomlop2010}
Tommasi, C. and L\'{o}pez-Fidalgo, J. (2010).
\newblock Bayesian optimum designs for discriminating between models with any
  distribution.
\newblock {\em Computational Statistics \& Data Analysis}, 54(1):143--150.

\bibitem[Ucinski and Bogacka, 2005]{ucibog2005}
Ucinski, D. and Bogacka, B. (2005).
\newblock {$T$}-optimum designs for discrimination between two multiresponse
  dynamic models.
\newblock {\em Journal of the Royal Statistical Society, Ser. B}, 67:3--18.

\bibitem[Wiens, 2009]{wiens2009}
Wiens, D.~P. (2009).
\newblock Robust discrimination designs, with {M}atlab code.
\newblock {\em Journal of the Royal Statistical Society, Ser.\ B}, 71:805--829.

\bibitem[Yang et~al., 2013]{yanbie2013}
Yang, M., Biedermann, S., and Tang, E. (2013).
\newblock On optimal designs for nonlinear models: A general and efficient
  algorithm.
\newblock {\em Journal of the American Statistical Association},
  108:1411--1420.

\bibitem[Yu, 2010]{yu2010}
Yu, Y. (2010).
\newblock Monotonic convergence of a general algorithm for computing optimal
  designs.
\newblock {\em The Annals of Statistics}, 38(3):1593--1606.

\end{thebibliography}
\end{small}
\end{document}